\documentclass{article}

\usepackage{arxiv}

\usepackage[utf8]{inputenc} 
\usepackage[T1]{fontenc}    
\usepackage{hyperref}       
\usepackage{url}            
\usepackage{booktabs}       
\usepackage{amsfonts}       
\usepackage{nicefrac}       
\usepackage{microtype}      
\usepackage{lipsum}
\usepackage{graphicx}

\usepackage{amsmath}
\usepackage{amssymb}
\usepackage{amsthm}

\theoremstyle{plain}
\newtheorem{theorem}{Theorem}[section]
\newtheorem{lemma}[theorem]{Lemma}
\newtheorem*{corollary}{Corollary}

\theoremstyle{definition}
\newtheorem{definition}{Definition}[section]

\theoremstyle{remark}
\newtheorem*{remark}{Remark}

\usepackage{algorithm}
\usepackage[noend]{algorithmic}

\newif\ifCOM
\COMfalse

\newcommand{\vectornorm}[1]{\|#1\|}
\begin{document}

\title{The algebraic structure of the densification and the sparsification tasks for CSPs
}


\author{Rustem Takhanov\\
  School of Sciences and Humanities\\
  Nazarbayev University\\
  Astana \\
  \texttt{rustem.takhanov@nu.edu.kz} \\
}


\date{Received: date / Accepted: date}

\maketitle

\begin{abstract}
The tractability of certain CSPs for dense or sparse instances is known from the 90s. Recently, the densification and the sparsification of CSPs were formulated as computational tasks and the systematical study of their computational complexity was initiated.

We approach this problem by introducing the densification operator, i.e. the closure operator that, given an instance of a CSP, outputs all constraints that are satisfied by all of its solutions. According to the Galois theory of closure operators, any such operator is related to a certain implicational system (or, a functional dependency) $\Sigma$. We are specifically interested in those classes of fixed-template CSPs, parameterized by constraint languages $\Gamma$, for which there is an implicational system $\Sigma$ whose size is a polynomial in the number of variables $n$. 
We show that in the Boolean case, such implicational systems exist if and only if $\Gamma$ is of bounded width. 
For such languages, $\Sigma$ can be computed in log-space or in a logarithmic time with a polynomial number of processors. Given an implicational system $\Sigma$, the densification task is equivalent to the computation of the closure of input constraints. The sparsification task is equivalent to the computation of the minimal key.  

\end{abstract}

\section{Introduction}
In the constraint satisfaction problem (CSP)~\cite{Feder,Bul-dich,Zhuk} we are given a set of variables with prescribed domains and a set of constraints. The task's goal is to assign each variable a value such that all the constraints are satisfied. Given an instance of CSP, besides the classical formulation, one can formulate many other tasks, such as maximum/minimum CSPs (Max/Min-CSPs)~\cite{Khanna}, valued CSP (VCSPs)~\cite{COOPER2003311,Schiex}, counting CSPs~\cite{Pesant05,Expressibility}, promise CSPs~\cite{Jakub,Brakensiek},  quantified CSPs~\cite{carvalho,Ferguson,Bauland2009TheCO}, and others.  
Thus, the computational task of finding a single solution is not the only aspect that is of interest from the perspective of applications of CSPs. 

Sometimes in applications we have a CSP instance that defines a set of solutions, and we need to preprocess the instance by making it denser (i.e. adding new constraints) or, vice versa, sparser (removing as many constraints as we can) without changing the set of solutions.
Let us give an example of such an application. 
The paper by Jia Deng et al.~\cite{DengJia} is dedicated to the Conditional Random Field (CRF) based on the so-called HEX graphs. 
The algorithm for the inference in CRFs presented there is based on the standard junction tree algorithm~\cite{Bishop}, but with one additional trick --- before constructing the junction tree of the factor graph, the factor tree is sparsified. 
This step aims to make the factor graph as close to the tree structure as possible. After that step, cliques of the junction tree are expected to have fewer nodes.
The sparsification of the HEX graph done in this approach is equivalent to the sparsification of a CSP instance, i.e. the deletion of as many constraints as possible while maintaining the set of solutions. The term “sparsification” is also used in a related line of work in which the goal is, given a CSP instance, to reduce the number of constraints without changing
the satisfiability of an instance~\cite{ChenJansen,Lagerkvist}. 

As was suggested in~\cite{DengJia}, the densification of a CSP instance could also help make inference algorithms more efficient.
If the factor tree is densified, then for every clique $c$ of the factor graph, the number of consistent assignments to variables of the clique $c$ is smaller. Thus, reducing the state space for each clique makes the inference faster. 
The sparsification-densification approach substantially accelerates the computation of the marginals as the number of nodes grows. 

It is well-known that the complexity of the sparsification problem, as well as the worst-case sparsifiability, depends on the constraint language, i.e. the types of constraints allowed in CSP. The computational complexity was completely classified for constraint languages consisting of the so-called irreducible relations~\cite{Schnoor}.

For a constraint language that consists of Boolean relations of the form $A_1 \wedge A_2  \wedge  ...  \wedge  A_n\to B$ (so-called pure Horn clauses),  the sparsification task is equivalent to the problem of finding a minimum size cover of a given functional dependency (FD) table. 
The last problem was studied in database theory long ago and is considered a classical topic. It was shown that this problem is NP-hard both in the general case and in the case a cover is restricted to be a subset of the given FD table. Surprisingly, if we re-define the size of a cover as the number of distinct left-hand side expressions $A_1 \wedge A_2  \wedge  ...  \wedge  A_n$, then the problem is polynomially solvable~\cite{Maier83}.

An important generalization of the previous constraint language is a set of Horn clauses (i.e. $B$ can be equal to ${\bf False}$). The sparsification problem for this language is known by the name {\em Horn  minimization}, i.e. it is a problem of finding the minimum size Horn formula that is equivalent to an input Horn formula.  Horn minimization is NP-hard if the number of clauses is to be minimized~\cite{Ausiello1984,Boros94onthe}, or if the number of literals is to be minimized~\cite{HAMMER1993131}. 
Moreover, in the former case Horn minimization cannot be $2^{\log ^{1-\epsilon} (n)}$-approximated if ${\rm NP}\not \subseteq {\rm DTIME}(n^{{\rm polylog}(n)})$~\cite{Bhattacharya}.

An example of a tractable sparsification problem is 2-SAT formula minimization~\cite{Chang2006HornFM} which corresponds to the constraint language of binary relations over the Boolean domain. 

The key idea of this paper's approach is to consider both  densification and  sparsification as two operations defined on the same set, i.e. the set of possible constraints. We observe that the densification is a closure operator on a finite set, and therefore, according to Galois theory~\cite{CASPARD2003241}, it can be defined using a functional dependency table, or so-called implicational system $\Sigma$ (over a set of possible constraints and, maybe, some additional literals). 
It turns out that $\Sigma$ can have a size bounded by some polynomial of the number of variables only if the constraint language is of bounded width (for tractable languages not of bounded width, the size of $\Sigma$ could still be substantially smaller than for NP-hard languages). For the Boolean domain, all languages of bounded width have a polynomial-size implicational system $\Sigma$. 

Given an implicational system $\Sigma$, the sparsification problem can be reformulated as a problem of finding the minimal key in $\Sigma$, i.e. such a set of constraints whose densification is the same as the densification of initial constraints. This task was actively studied in database theory, and we exploit the standard algorithm for the solution of the minimal key problem, found by Luchessi and Osborn~\cite{LUCCHESI1978270}. If $|\Sigma| = {\mathcal O}({\rm poly}(n))$ and literals of $\Sigma$ are all from the set of possible constraints, this leads us to a ${\mathcal O}({\rm poly}(n)\cdot N^2)$-sparsification algorithm where $N$ is the number of non-redundant sparsifications of an input instance. This algorithm can be applied to the Horn minimization problem, and, to our knowledge, this is the first algorithm that is polynomial in $N$. Of course, in the worst-case $N$ is large. 

Besides the mentioned works, densification/sparsification tasks were also studied for soft CSPs, and this unrelated research direction includes graph densification~\cite{Hardt,Escolano,Curado}, binary CSP sparsification~\cite{Benczur,Batson,Andoni,Filtser,Butti} and spectral sparsification of graphs and hypergraphs~\cite{FungWaiShing,HypergraphsSparsification}. In the 90's it was found that dense CSP instances (i.e. when the number of $k$-ary constraints is $\Theta(n^k)$) admit efficient algorithms for the Max-$k$-CSP and the maximum assignment problems~\cite{DenseInstances,newrounding,regularitylemma}. 
Though we deal with crisp CSPs and not any CSP instance can be densified to $\Theta(n^k)$ constraints, the idea to densify a CSP instance seems natural in this context. 
Note that the densification of a CSP that we study in our paper is substantially different from the notion of the densification of a graph. Initially, Hardt et al.~\cite{Hardt} define the densification of the graph $G=(V,E)$ as a new graph $H=(V, E'), E'\supseteq E$ such that the cardinalities of cuts in $G$ and $H$ are proportional. In~\cite{Escolano,Curado} and in the Ph.D. Thesis~\cite{CuradoPhD} the densification was naturally applied in a clustering problem to neighborhood graphs to make more intra-class links and smaller overhead of inter-class links.  It was shown that this makes the Laplacian of a graph better conditioned for a subsequent application of spectral methods. A theoretical analysis of the densification/sparsification tasks for soft CSPs requires mathematical tools substantially different from those that we develop in the paper.

\section{Preliminaries}\label{sec:prelim}
We assume that ${\rm P} \neq {\rm NP}$. The set $\{1,...,k\}$ is denoted by $[k]$. Given a relation $\varrho\subseteq R^{s}$ and a tuple $\mathbf{a}\in R^{s'}$, by $\vectornorm{\varrho}$ and $|\mathbf{a}|$ we denote $s$ and $s'$, respectively. 
A relational structure is a tuple $\mathbf{R} = (R, r_1, ..., r_k)$ where $R$ is finite set, called the domain of $\mathbf{R}$, and $r_i\subseteq R^{\vectornorm{r_i}}$, $i\in [k]$.
 If $p_0\in [\vectornorm{\varrho}]$, then ${\rm pr}_{\{p_0\}} (\varrho) = \{a_{p_0}| (a_1, ..., a_k)\in \varrho\}$, if $p_0<p_1 \leq \vectornorm{\varrho}$, then ${\rm pr}_{\{p_0,p_1\}} (\varrho) = \{(a_{p_0}, a_{p_1})|  (a_1, ..., a_k)\in \varrho\}$ etc.

\subsection{The homomorphism formulation of CSP}
Let us define first the notion of a homomorphism between relational structures.

\begin{definition} Let $\mathbf{R} = (V, r_1, ..., r_s)$ and $\mathbf{R}' = (V', r'_1, ..., r'_s)$ be relational structures with a common signature (that is arities of $r_i$ and $r'_i$ are the same for every $i \in [s]$). A mapping $h\colon V \to V'$ is called a \emph{homomorphism} from $\mathbf{R}$ to $\mathbf{R}'$ if for every $i \in [s]$ and for any $(x_1, ..., x_{\vectornorm{r_i}}) \in r_i$ we have that $\big((h(x_1), ..., h(x_{\vectornorm{r'_i}})\big) \in r'_i$. The set of all homomorphisms from $\mathbf{R}$ to $\mathbf{R}'$  is denoted by ${\rm Hom}(\mathbf{R}, \mathbf{R}')$.
\end{definition}

The classical CSP can be formulated as a homomorphism problem. 
\begin{definition} The {\bf CSP}  is a search task with:
\begin{itemize}
\item {\bf An instance:} two relational structures with a common signature, $\mathbf{R} = (V,r_1, ..., r_s)$ and $\boldsymbol{\Gamma} = (D, \varrho_1, ..., \varrho_s)$.
\item {\bf An output:} a homomorphism $h:\mathbf{R} \to \boldsymbol{\Gamma}$ if it exists, or answer ${\rm None}$, if it does not exist.
\end{itemize}
\end{definition}

A finite relational structure $\boldsymbol{\Gamma} = (D, \varrho_1, ..., \varrho_s)$ over a fixed finite domain $D$ is sometimes called a template.
For such $\boldsymbol{\Gamma}$ we will denote by $\Gamma$ (without boldface) the set of relations $\{ \varrho_1, ..., \varrho_s \}$. The set $\Gamma$ is called the constraint language. 

\begin{definition} The {\bf fixed template CSP} for a given template $\boldsymbol{\Gamma} = (D, \varrho_1, ..., \varrho_s)$, denoted ${\rm CSP}(\boldsymbol{\Gamma})$, is defined as follows:
given a relational structure $\mathbf{R} = (V,r_1, ..., r_s)$ of the same signature as $\boldsymbol{\Gamma}$, solve the CSP for an instance $(\mathbf{R}, \boldsymbol{\Gamma})$. If ${\rm CSP}(\boldsymbol{\Gamma})$ is solvable in a polynomial time, then $\boldsymbol{\Gamma}$ is called tractable. Otherwise, $\boldsymbol{\Gamma}$ is called NP-hard~\cite{Bul-dich,Zhuk}.
\end{definition}

\subsection{Algebraic approach to CSPs}
In the paper we will need standard definitions of primitive positive formulas and polymorphisms.
\begin{definition}
Let $\tau=\{\pi_1, ..., \pi_s\}$ be a set of symbols for predicates, with the arity $n_i$ assigned to $\pi_i$.
A first-order formula $\Phi(x_1, ..., x_k) = \exists x_{k+1}... x_{n} \Xi(x_1, ..., x_n)$ where $\Xi(x_1, ..., x_n) = \bigwedge_{t=1}^N \pi_{j_t} (x_{o_{t1}}, x_{o_{t2}}, ..., x_{o_{tn_{j_t}}})$, $j_t\in [s]$, $o_{tq}\in [n]$ is called the primitive positive formula over the vocabulary $\tau$. For a relational structure $\mathbf{R} = (V,r_1, ..., r_s)$, $\vectornorm{r_i}=n_i, i\in [s]$, $\Phi^{\mathbf{R}}$ denotes a $k$-ary predicate $$\{(a_1, ..., a_k)|a_i\in V, i\in [k],  \exists a_{k+1}, \cdots, a_n\in V:  (a_{o_{t1}}, a_{o_{t2}}, ..., a_{o_{tn_{j_t}}})\in r_{j_t}, t\in [N]\},$$ i.e. the result of interpreting the formula $\Phi$ on the model $\mathbf{R}$, where $\pi_i$ is interpreted as $r_i$. 
\end{definition}
For $\boldsymbol{\Gamma} = (D,\varrho_1, ..., \varrho_s)$ and $\tau =\{\pi_1, ..., \pi_s\}$, let us denote the set $\{ \Psi^{\boldsymbol{\Gamma}} | \Psi {\rm \,\,is\,\,primitive \,\,positive}\\ {\rm formula\,\,over\,\,}\tau \}$ by $\langle \Gamma\rangle$. 

\begin{definition}\label{D:first}
Let $\varrho  \subseteq D^m $ and $f:D^n  \to D$.
We say that the predicate $\varrho$ is preserved by $f$ (or, $f$ is a polymorphism of $\varrho$) if,
for every $\left( {x_1^i ,...,x_m^i } \right) \in \varrho, 1 \leq i \leq n$, we have that $\left( {f\left( {x_1^1 ,...,x_1^n
} \right),...,f\left( {x_m^1 ,...,x_m^n } \right)} \right) \in
\varrho $.
\end{definition}

For a set of predicates $\Gamma\subseteq \{\varrho| \varrho\subseteq D^m\}$, let ${\rm Pol}\left( \Gamma \right)$ denote the
set of operations $f: D^n \rightarrow D$ such that $f$ is a polymorphism of all predicates in $\Gamma$.  For a set of
operations $F\subseteq \{f| f: D^n\rightarrow D\}$, let $ {\rm Inv}\left( F \right)$ denote the set of
predicates $\varrho\subseteq D^m$ preserved under the operations of $F$. The next result is well-known \cite{bodnarchuk,geiger}.
\begin{theorem}[Geiger, Bodnarchuk, Kaluznin, Kotov, Romov]\label{T:1}
For a set of predicates $\Gamma$ over a finite set $D$, $\langle \Gamma \rangle = {\rm Inv}\left( {\rm Pol}\left( \Gamma \right) \right)$.
\end{theorem}
It is well-known that the computational complexity of fixed-template CSPs, counting CSPs, VCSPs etc. is determined by the closure $\langle \Gamma\rangle$, and therefore, by the corresponding functional clone ${\rm Pol}\left( \Gamma \right)$. 

\section {The fixed template densification and sparsification problems}
Let us give a general definition of maximality and list some properties of maximal instances.
\begin{definition}\label{maximality} An instance $(\mathbf{R}, \boldsymbol{\Gamma})$ of CSP, where $\mathbf{R} = (V,r_1, ..., r_s)$ and $\boldsymbol{\Gamma} = (D, \varrho_1, ..., \varrho_s)$, is said to be maximal if for any $\mathbf{R}' = (V,r'_1, ..., r'_s)$ such that $r'_i\supseteq r_i$, $i\in [s]$ we have ${\rm Hom}(\mathbf{R}, \boldsymbol{\Gamma}) \ne {\rm Hom}(\mathbf{R}', \boldsymbol{\Gamma})$, unless $\mathbf{R}' = \mathbf{R}$.
\end{definition}

The following characterization of maximal instances is evident from Definition~\ref{maximality} (also, see Theorem 1 in~\cite{Takhanov2007}).
\begin{theorem}\label{jvm}
An instance $(\mathbf{R} = (V,r_1, ..., r_s), \boldsymbol{\Gamma} = (D, \varrho_1, ..., \varrho_s))$ is maximal if and only if for any $i\in [s]$ and any $(v_1, ..., v_{\vectornorm{r_i}})\notin r_i$ there exists $h\in {\rm Hom}(\mathbf{R}, \boldsymbol{\Gamma})$ such that $(h(v_1), ..., h(v_{\vectornorm{r_i}}))\notin \varrho_i$.
\end{theorem}

One can prove the following simple existence theorem (Statement 1 in~\cite{Takhanov2007}).
\begin{theorem}\label{old-mine}
For any instance $(\mathbf{R} = (V,r_1, ..., r_s), \boldsymbol{\Gamma} = (D, \varrho_1, ..., \varrho_s))$ of CSP, there exists a unique maximal instance $(\mathbf{R}'= (V,r'_1, ..., r'_s), \boldsymbol{\Gamma})$ such that $r'_i\supseteq r_i, i\in [s]$
 and ${\rm Hom}(\mathbf{R}, \boldsymbol{\Gamma}) = {\rm Hom}(\mathbf{R}', \boldsymbol{\Gamma})$. Moreover, if ${\rm Hom}(\mathbf{R}, \boldsymbol{\Gamma})\ne \emptyset$, then
$$r'_i=\bigcap_{h\in {\rm Hom}(\mathbf{R}, \boldsymbol{\Gamma})} h^{-1}(\varrho_i), i\in [s]$$
\end{theorem}

Thus, the maximal instance $(\mathbf{R}', \boldsymbol{\Gamma})$ from Theorem~\ref{old-mine} can be called the densification of $(\mathbf{R}, \boldsymbol{\Gamma})$. Let us now formulate constructing $(\mathbf{R}', \boldsymbol{\Gamma})$ from $(\mathbf{R}, \boldsymbol{\Gamma})$ as an algorithmic problem.

\begin{definition}\label{dense-define} The {\bf densification problem}, denoted ${\rm Dense}$,  is a search task with:
\begin{itemize}
\item {\bf An instance:} two relational structures with a common signature, $\mathbf{R} = (V,r_1, ..., r_s)$ and $\boldsymbol{\Gamma} = (D, \varrho_1, ..., \varrho_s)$.
\item {\bf An output:} a maximal instance $(\mathbf{R}'= (V,r'_1, ..., r'_s), \boldsymbol{\Gamma})$ such that $r'_i\supseteq r_i, i\in [s]$ and ${\rm Hom}(\mathbf{R}, \boldsymbol{\Gamma}) = {\rm Hom}(\mathbf{R}', \boldsymbol{\Gamma})$.
\end{itemize}
Also, let $D$ be a finite set and $\boldsymbol{\Gamma}$ a relational structure with a domain $D$. Then, the {\bf fixed template densification problem} for the template $\boldsymbol{\Gamma}$, denoted ${\rm Dense}(\boldsymbol{\Gamma})$, is defined as follows:
given a relational structure $\mathbf{R} = (V,r_1, ..., r_s)$ of the same signature as $\boldsymbol{\Gamma}$, solve the densification problem for an instance $(\mathbf{R}, \boldsymbol{\Gamma})$.
\end{definition}

Let $\Gamma = \{ \varrho_1, \cdots, \varrho_s \}$. The language $\Gamma$ is called constant-preserving if there is $a\in D$ such that $(a, \cdots, a)\in \varrho_i$ for any $i\in [s]$.
For a pair $(\mathbf{R}, \boldsymbol{\Gamma})$, where $\Gamma$ is not a constant-preserving language, the corresponding densification is non-trivial, i.e. $\mathbf{R}'\ne (V, V^{\vectornorm{r_1}}, \cdots, V^{\vectornorm{r_s}})$, if and only if ${\rm Hom}(\mathbf{R}, \boldsymbol{\Gamma})\ne \emptyset$. Therefore, the densification problem for such templates $\boldsymbol{\Gamma}$ is at least as hard as the decision form of CSP. In other words, if the decision form of ${\rm CSP}(\boldsymbol{\Gamma})$ is NP-hard (which is known to be polynomially equivalent to the search form), then all the more ${\rm Dense}(\boldsymbol{\Gamma})$ is NP-hard.

For a Boolean constraint language $\Gamma$, we say that $\Gamma$ is Schaefer in one of the following cases: 1) $x\vee y\in {\rm Pol}(\Gamma)$, 2) $x\wedge y\in {\rm Pol}(\Gamma)$, 3)  $x\oplus y\oplus z\in {\rm Pol}(\Gamma)$, 4) ${\rm mjy}(x,y,z) = (x\wedge y)\vee (x\wedge y)\vee (x\wedge z)\in {\rm Pol}(\Gamma)$. 
The complexity of ${\rm Dense}(\boldsymbol{\Gamma})$ in the Boolean case can be simply described by the following theorem whose proof uses earlier results of~\cite{Schnoor2008} and~\cite{Hemaspaandra}.
For completeness, a detailed proof can be found in Section~\ref{Boolean-dense}.
\begin{theorem}\label{Boolean-dense-thm} For $D = \{0,1\}$, ${\rm Dense}(\boldsymbol{\Gamma})$ is polynomially solvable if and only if $\Gamma$ is Schaefer. 
\end{theorem}

Let us introduce the sparsification problem.

\begin{definition}\label{minimality} An instance $(\mathbf{R}, \boldsymbol{\Gamma})$ of CSP, where $\mathbf{R} = (V,r_1, ..., r_s)$ and $\boldsymbol{\Gamma} = (D, \varrho_1, ..., \varrho_s)$, is said to be minimal if for any $\mathbf{T} = (V,t_1, ..., t_s)$ such that $t_i\subseteq r_i, i\in [s]$  we have ${\rm Hom} (\mathbf{R}, \boldsymbol{\Gamma})\ne {\rm Hom} (\mathbf{T}, \boldsymbol{\Gamma})$, unless $\mathbf{T} = \mathbf{R}$.
\end{definition}
Let us define:
\begin{equation}
\begin{split}
{\rm Min} (\mathbf{R}, \boldsymbol{\Gamma}) = \big\{\mathbf{R}'=(V,r'_1, ..., r'_s) \mid
{\rm Hom} (\mathbf{R}, \boldsymbol{\Gamma})={\rm Hom} (\mathbf{R}', \boldsymbol{\Gamma}), (\mathbf{R}', \boldsymbol{\Gamma})\rm{\,\,is\,\,minimal}\big\}
\end{split}
\end{equation}

\begin{definition} The {\bf sparsification problem}, denoted ${\rm Sparse}$,  is a search task with:
\begin{itemize}
\item {\bf An instance:} two relational structures with a common signature, $\mathbf{R} = (V,r_1, ..., r_s)$ and $\boldsymbol{\Gamma} = (D, \varrho_1, ..., \varrho_s)$.
\item {\bf An output:} List of all elements of ${\rm Min} (\mathbf{R}, \boldsymbol{\Gamma})$.
\end{itemize}
Also, let $D$ be a finite set and $\boldsymbol{\Gamma}$ a relational structure with a domain $D$. Then, the {\bf fixed template sparsification problem} for the template $\boldsymbol{\Gamma}$, denoted ${\rm Sparse}(\boldsymbol{\Gamma})$, is defined as follows:
given a relational structure $\mathbf{R} = (V,r_1, ..., r_s)$ of the same signature as $\boldsymbol{\Gamma}$, solve the sparsification problem for an instance $(\mathbf{R}, \boldsymbol{\Gamma})$.
\end{definition}

\begin{remark} In many aplications  ${\rm Min} (\mathbf{R}, \boldsymbol{\Gamma})$ is of moderate size, though potentially it can depend on $|V|$ exponentially. Also, $\mathbf{R}'=(V,r'_1, ..., r'_s) \in {\rm Min} (\mathbf{R}, \boldsymbol{\Gamma})$ is not necessarily a substructure of $\mathbf{R}$, i.e. it is possible that $r'_i\not\subseteq r_i$. Enforcing $r'_i\subseteq r_i, i\in [s]$ in the definition of ${\rm Min} (\mathbf{R}, \boldsymbol{\Gamma})$ is discussed in Remark~\ref{Gottlob}.
\end{remark}

\section {Densification as the closure operator}

Let us introduce a set of all possible constraints over $\Gamma$ on a set of variables $V$:
$$
\mathcal{C}^{\boldsymbol{\Gamma}}_V = \{\langle (v_1, ..., v_{\vectornorm{\varrho_i}}), \varrho_i\rangle | i\in [s], v_1, ..., v_{\vectornorm{\varrho_i}}\in V\}
$$
Any instance of ${\rm CSP}(\boldsymbol{\Gamma})$, a relational structure $\mathbf{R}=(V,r_1, ..., r_s)$, induces the following subset of $\mathcal{C}^{\boldsymbol{\Gamma}}_V$: 
$$
\mathcal{C}_\mathbf{R} = \{\langle (v_1, ..., v_{\vectornorm{\varrho_i}}), \varrho_i\rangle | i\in [s], (v_1, ..., v_{\vectornorm{\varrho_i}})\in r_i\} 
$$
Using that notation, the densification can be understood as an operator ${\rm Dense}: 2^{\mathcal{C}^{\boldsymbol{\Gamma}}_V}\rightarrow 2^{\mathcal{C}^{\boldsymbol{\Gamma}}_V}$ such that:
\begin{equation*}
\begin{split}
{\rm Dense} (\mathcal{C}_\mathbf{R}) = \big\{\langle (v_1, ..., v_{\vectornorm{\varrho_i}}), \varrho_i\rangle |  i\in [s],  (v_1, ..., v_{\vectornorm{\varrho_i}})\in \bigcap_{h\in {\rm Hom}(\mathbf{R}, \boldsymbol{\Gamma})} h^{-1}(\varrho_i) \big\}
\end{split}
\end{equation*}
Thus, in the densification process we start from a set of constraints $\mathcal{C}_\mathbf{R}$ and simply add possible constraints to ${\rm Dense} (\mathcal{C}_\mathbf{R})$ while the set of solutions is preserved. Let us also define ${\rm Dense} (\mathcal{C}_\mathbf{R}) = \mathcal{C}^{\boldsymbol{\Gamma}}_V$ if 
${\rm Hom}(\mathbf{R}, \boldsymbol{\Gamma}) = \emptyset$. The densification operator satisfies the following conditions:
\begin{itemize}
\item  ${\rm Dense} (\mathcal{C}_\mathbf{R})\supseteq \mathcal{C}_\mathbf{R}$ (extensive)
\item ${\rm Dense} ({\rm Dense} (\mathcal{C}_\mathbf{R}))={\rm Dense} (\mathcal{C}_\mathbf{R})$ (idempotent)
\item $\mathcal{C}_{\mathbf{R}'}\subseteq \mathcal{C}_\mathbf{R} \Rightarrow {\rm Dense} (\mathcal{C}_{\mathbf{R}'})\subseteq {\rm Dense} (\mathcal{C}_\mathbf{R})$ (isotone)
\end{itemize}
Operators that satisfy these three conditions play the central role in universal algebra and are called the closure operators. There exists a duality between closure operators $o:2^S\rightarrow 2^S$ on a finite set $S$ and the so-called implicational systems (or functional dependencies) on $S$. Let us briefly describe this duality (details can be found in~\cite{CASPARD2003241}).

\begin{definition}\label{armstrong}
Let $S$ be a finite set. An implicational system $\Sigma$ on $S$ is a binary
relation $\Sigma\subseteq 2^S \times 2^S$. If $(A, B)\in \Sigma$, we write $A\rightarrow B$. A full implicational system on $S$ is an implicational system satisfying the three
following properties:
\begin{itemize}
\item  $A\rightarrow B, B\rightarrow C$  imply  $A\rightarrow C$ 
\item $A\subseteq B$  imply  $B\rightarrow A$
\item $A\rightarrow B$ and $C\rightarrow D$ imply $A\cup C \rightarrow B \cup D$.
\end{itemize}
\end{definition}
Any implicational system $\Sigma\subseteq 2^S \times 2^S$ has a minimal superset $\Sigma'\supseteq \Sigma$ that itself is a full implicational system on $S$. This system is called the closure of $\Sigma$ and is denoted by $\Sigma^\triangleright$. Let us call $\Sigma_1$ a  cover of $\Sigma_2$ if $\Sigma^\triangleright_1= \Sigma^\triangleright_2$.
\begin{theorem}[p. 264 \cite{CASPARD2003241}]\label{implicational}
Any implicational system $\Sigma\subseteq 2^S \times 2^S$ defines the closure operator $o:2^S\rightarrow 2^S$ by $o(A) = \{x\in S| A\rightarrow \{x\} \in \Sigma^\triangleright\}$.
Any closure operator $o:2^S\rightarrow 2^S$ on a finite set $S$ defines the full implicational system by $\{A\rightarrow B| B\subseteq o(A)\}$.
\end{theorem}

From Theorem~\ref{implicational} we obtain that the densification operator ${\rm Dense}: 2^{\mathcal{C}^{\boldsymbol{\Gamma}}_V}\rightarrow 2^{\mathcal{C}^{\boldsymbol{\Gamma}}_V}$ also corresponds to some full implicational system $\Sigma_V^{\boldsymbol{\Gamma}}\subseteq 2^{\mathcal{C}^{\boldsymbol{\Gamma}}_V} \times 2^{\mathcal{C}^{\boldsymbol{\Gamma}}_V}$. 
Note that the system $\Sigma_V^{\boldsymbol{\Gamma}}$ depends only on the set $V$ and the template $\boldsymbol{\Gamma}$, but does not depend on relations $r_i, i\in [s]$ of the relational structure $\mathbf{R}$. 

Note the densification problem described in Definition~\ref{dense-define}  can be understood as a computation of the monotone function  ${\rm Dense}:\cup_{n=1}^\infty 2^{\mathcal{C}_{[n]}^{\boldsymbol{\Gamma}} } \to \cup_{n=1}^\infty  2^{\mathcal{C}_{[n]}^{\boldsymbol{\Gamma}} }$. With a little abuse of terminology, let us define the class mP/poly as a class of monotone functions $M:\bigcup_{i=0}^\infty\{0,1\}^{n_i}\to \bigcup_{i=0}^\infty\{0,1\}^{m_i}$ for which $M(\{0,1\}^{n_i})\subseteq \{0,1\}^{m_i}$ and $M|_{\{0,1\}^{n_i}}$ can be computed by a circuit of size ${\rm poly}(n_i)$ that uses only $\vee$ and $\wedge$ in gates. Thus, ${\rm Dense}(\boldsymbol{\Gamma})\in {\rm mP/poly}$ denotes the fact that the corresponding densification operator is in mP/poly.

\section{The polynomial densification operator}
Let denote $\Sigma_n^{\boldsymbol{\Gamma}} = \Sigma_{[n]}^{\boldsymbol{\Gamma}}$. 
The most general languages with a kind of polynomial densification operator can be described as follows.
\begin{definition}\label{dense-poly2} The template $\boldsymbol{\Gamma}$ is said to have a weak polynomial densification operator, if  for any $n\in {\mathbb N}$ there exists an implicational system $\Sigma$ on $S\supseteq {\mathcal C}^{\boldsymbol{\Gamma}}_{n}$ of size $|\Sigma| = {\mathcal O}({\rm poly}(n))$ that acts on ${\mathcal C}^{\boldsymbol{\Gamma}}_{n}$ as the densification operator, i.e.  $\Sigma_n^{\boldsymbol{\Gamma}} = \{(A\rightarrow B)\in \Sigma^\triangleright| A,B\subseteq {\mathcal C}^{\boldsymbol{\Gamma}}_{n}\}$.
\end{definition}
Using database theory language~\cite{Levene}, the last definition describes such languages $\Gamma$ for which there exists an implicational system of polynomial size whose projection on ${\mathcal C}^{\boldsymbol{\Gamma}}_{n}$ coincides with $\Sigma_n^{\boldsymbol{\Gamma}}$.
Note that in Definition~\ref{dense-poly2}, a weak densification operator acts on a wider set than ${\mathcal C}^{\boldsymbol{\Gamma}}_{n}$: an addition of new literals to ${\mathcal C}^{\boldsymbol{\Gamma}}_{n}$, sometimes, allows to substantially simplify a set of implications~\cite{FISCHER1983323}. Though we are not aware of an example of a language $\Gamma$ for which any cover $\Sigma\subseteq \Sigma_n^{\boldsymbol{\Gamma}}$ of $\Sigma_n^{\boldsymbol{\Gamma}}$ is exponential in size, but still $\Gamma$ has a weak polynomial densification operator.

\section{Main results}\label{main-res}
Recall that bounded width languages are languages for which $\neg {\rm CSP}(\boldsymbol{\Gamma})$ can be recognized by a Datalog program~\cite{Feder}. Concerning the weak polynomial densification, we obtain the following result.
\begin{theorem}\label{main-bounded} For the general domain $D$, if $\boldsymbol{\Gamma}$ has a weak polynomial densification operator, then $\boldsymbol{\Gamma}$ is of bounded width. For the Boolean case, $D = \{0,1\}$, $\boldsymbol{\Gamma}$ has a weak polynomial densification operator if and only if ${\rm Pol}(\Gamma)$ contains either $\vee$, or $\wedge$, or ${\rm mjy}(x,y,z)$.
\end{theorem}
The first part of the latter theorem is proved in Section~\ref{strong-dense} and the Boolean case is considered in Section~\ref{bounded-width-main}. 
We also prove the following statement for the sparsification problem (Section~\ref{DS-section}).
\begin{theorem}\label{main-sparse} If $\Sigma\subseteq \Sigma^{\boldsymbol{\Gamma}}_V$ is a cover of $\Sigma^{\boldsymbol{\Gamma}}_V$ that can be computed in time ${\rm poly}(|V|)$, then given an instance ${\mathbf R} = (V, r_1, ..., r_s)$ of ${\rm Sparse}(\boldsymbol{\Gamma})$, all elements of ${\rm Min} (\mathbf{R}, \boldsymbol{\Gamma})$ can be listed in time ${\mathcal O} ({\rm poly}(|V|)\cdot |{\rm Min} (\mathbf{R}, \boldsymbol{\Gamma})|^2)$.
\end{theorem}

\section{Weak polynomial densification implies bounded width}\label{strong-dense}
A set of languages with a weak polynomial densification operator turns out to be a subset  of a set of languages of bounded width. Below we demonstrate this fact in two steps. First, we prove that from a weak polynomial densification operator one can construct a polynomial-size monotone circuit that computes $\neg{\rm CSP}(\boldsymbol{\Gamma})$. Further, we exploit a well-known result from a theory of fixed-template CSPs connecting the bounded width with such circuits.
\begin{theorem}\label{no-linear} If
$\boldsymbol{\Gamma}$ has a weak polynomial densification operator, then the decision version of $\neg{\rm CSP}(\boldsymbol{\Gamma})$ can be computed by a polynomial-size monotone circuit.
\end{theorem}

\begin{proof}If $\boldsymbol{\Gamma}$ is constant-preserving, then $\neg{\rm CSP}(\boldsymbol{\Gamma})$ is trivial, i.e. we can assume that $\boldsymbol{\Gamma}$ is not constant-preserving.
Let $\Sigma_n$ be an implicational system on $S_n\supseteq {\mathcal C}^{\boldsymbol{\Gamma}}_{n}$ such that $\Sigma_n^\triangleright \cap (2^{{\mathcal C}^{\boldsymbol{\Gamma}}_{n}})^2 = \Sigma_n^{\boldsymbol{\Gamma}}$ and $|\Sigma_n| = \mathcal{O}({\rm poly}(n))$. We can assume that $S_n =  \mathcal{O}({\rm poly}(n))$ and every rule in $\Sigma_n$ has a form $A\to x$, $x\in S_n$. Let ${\mathbf R}$ be an instance of ${\rm CSP}(\boldsymbol{\Gamma})$ and $x\in {\mathcal C}^{\boldsymbol{\Gamma}}_{n}$. The rule ${\mathcal C}_{\mathbf R}\to x$ is in $\Sigma_n^\triangleright$ if and only if $x$ is derivable from ${\mathcal C}_{\mathbf R}$ using implications from $\Sigma_n$.
Formally, the latter means that there is a directed acyclic graph $T = (U,E)$ with a labeling function $l: U\to S_n$ such that: (a) there is only one element with no outcoming edges, the root $r\in U$, and it is labeled by $x$, i.e. $l(r)=x$, (b) every node with no incoming edges is labeled by an element of ${\mathcal C}_{\mathbf R}$, (c) if a node $v\in U$ has incoming edges $(c_1, v), ..., (c_{d(v)}, v)$, then $(\{l(c_1), ..., l(c_{d(v)})\}\to l(v))\in \Sigma_n$. Moreover, the depth of $T$ is bounded by $|S_n|$, because $x$ can be derived from ${\mathcal C}_{\mathbf R}$ in no more than $|S_n|$ steps if no attribute is derived twice. 

Consider a monotone circuit $M$ whose set of variables, denoted by $W$, consists of $|S_n|$ layers $U_1, ..., U_{|S_n|}$ such that $i$-th layer is a set of variables $v_{i,a}, a\in S_{n}$. For any rule $b\in S_n$ and every $i\in [|S_n|-1]$ there is a monotone logic gate $$v_{i+1, b} = v_{i, b} \vee \bigvee_{(\{a_1, ..., a_l\}\to b)\in \Sigma_n} (v_{i, a_1}\wedge v_{i, a_2}\wedge ... \wedge v_{i, a_l})$$ that computes the value of $v_{i+1, b}$ from inputs of the previous layer.

Any instance ${\mathbf R}$ can be encoded as a Boolean vector ${\mathbf v}_{\mathbf R}\in \{0,1\}^{S_{n}}$ such that ${\mathbf v}_{\mathbf R}(x)=1$ if and only if $x\in {\mathcal C}_{\mathbf R}$. If we set input variables of $M$ to ${\mathbf v}_{\mathbf R}$, i.e. $v_{1,a} := {\mathbf v}_{\mathbf R}(a), a\in S_{n}$, then output variables of $M$, i.e.  $v_{|S_n|,a}, a\in S_{n}$, will satisfy: for any $x\in {\mathcal C}^{\boldsymbol{\Gamma}}_{n}$, $v_{|S_n|,x}=1$ if and only if $({\mathcal C}_{\mathbf R}\to x)\in \Sigma_n^\triangleright$. Let us briefly outline the proof of the last statement.

Indeed, let $v_{|S_n|,x}=1$, $x\in {\mathcal C}_{\mathbf R}$.  For any variable $v_{i,b}\in W$ such that $v_{i,b}=1$ let us define ${\rm early}(v_{i,b}) = v_{i',b}$ where $v_{i',b}=1$ and $v_{i'-1,b}=0$ and ${\rm source}(v_{i,b}) = \{v_{i'-1, a_1}, v_{i'-1, a_2}, ... , v_{i'-1, a_l}\}$ if $(\{a_1, ..., a_l\}\to b)\in \Sigma_n$ and $v_{i'-1, a_1} = 1, v_{i'-1, a_2} = 1, ... , v_{i'-1, a_l}=1$.
Then, a rooted directed acyclic graph $T_x = (U,E)$ with a labeling $l: U\to S_n$ can be constructed  by defining $U = \{{\rm early}(v_{i,b})|v_{i,b}\in W, v_{i,b}=1\}$ and $l({\rm early}(v_{i,b}))=b$. Edges of $T_x$ are defined in the following way: if $v_{i',b} = {\rm early}(v_{i,b})$ and $v_{i',b}$ was assigned to 1 by the gate $v_{i', b} = v_{i'-1, b} \vee (v_{i'-1, a_1}\wedge v_{i'-1, a_2}\wedge ... \wedge v_{i'-1, a_l})\vee \cdots$ where ${\rm source}(v_{i,b}) = \{v_{i'-1, a_1}, v_{i'-1, a_2}, ... , v_{i'-1, a_l}\}$, then we connect nodes ${\rm early}(v_{i'-1, a_2}), ..., {\rm early}(v_{i'-1, a_l})$ to $v_{i',b}$ by incoming edges. It is easy to see that $T_x$ will satisfy properties (a), (b), (c) listed above. The opposite is also true, if there is a directed acyclic graph with a root $x$ that satisfies the properties (a), (b), (c), then $v_{|S_n|,x}=1$.

Thus, the expression $o = \bigwedge_{x\in {\mathcal C}^{\boldsymbol{\Gamma}}_n} v_{|S_n|,x}$ equals 1 if and only if $({\mathcal C}_{\mathbf R}\to {\mathcal C}_n^{\boldsymbol{\Gamma}}) \in \Sigma^{\boldsymbol{\Gamma}}_n$. Since $\boldsymbol{\Gamma}$ is not constant-preserving, the last means ${\rm Hom}({\mathbf R}, \boldsymbol{\Gamma}) = \emptyset$. Thus, ${\rm Hom}({\mathbf R}, \boldsymbol{\Gamma}) = \emptyset$ was computed by the polynomial-size monotone circuit $M$ (with an additional gate).
\end{proof}
The core of $\Gamma = \{\varrho_1, ..., \varrho_s\}$ is defined as ${\rm core}(\Gamma) = \{\varrho_1\cap g(D)^{n_1}, ..., \varrho_s\cap g(D)^{n_s}\}$, the constraint language over $g(D)$, where $g \in {\rm Hom}(\boldsymbol{\Gamma}, \boldsymbol{\Gamma})$ is such that $g(x) = g(g(x))$ and $$|g(D)| = \min\limits_{h\in {\rm Hom}(\boldsymbol{\Gamma}, \boldsymbol{\Gamma})} |h(D)|.$$
\begin{corollary}\label{corol-mon}
If $\boldsymbol{\Gamma}$ has a weak polynomial densification operator, then ${\rm core}(\Gamma)$ is of bounded width.
\end{corollary}
\begin{proof}
If $\boldsymbol{\Gamma}$ has a weak polynomial densification operator, then by Theorem~\ref{no-linear} $\neg{\rm CSP}(\Gamma)$ can be solved by a polynomial-size monotone circuit. Therefore, $\neg{\rm CSP}(\Gamma')$ where $\Gamma' = {\rm core}(\Gamma)\cup \{\{(a)\}| a\in g(D)\}$ can also be solved by a polynomial-size monotone circuit. We can use the standard reduction of $\neg{\rm CSP}(\Gamma')$ to $\neg{\rm CSP}({\rm core}(\Gamma)\cup \{\rho\})$ where $\rho\in \langle {\rm core}(\Gamma) \rangle$ is defined as $\{\langle \pi(a) \rangle_{a\in g(D)}| \pi: g(D)\to g(D), \pi\in {\rm pol}({\rm core}(\Gamma))\}$. 

The algebra ${\mathbb A}_{\Gamma'} = (g(D), {\rm pol}(\Gamma'))$ generates the variety of algebras ${\textit var}({\mathbb A}_{\Gamma'})$ (in the sense of Birkhoff's HSP theorem). Proposition 5.1. from~\cite{ability_to_count} states that if $\neg{\rm CSP}(\Gamma')$ can be computed by a polynomial-size monotone circuit, then ${\textit var}({\mathbb A}_{\Gamma'})$ omits both the unary and the affine type. According to a well-known result~\cite{Bulatov10boundedrelational,BartoKozik} this is equivalent to stating that $\Gamma'$ is of bounded width.
\end{proof}

\section{Algebraic approach to the classification of languages with a polynomial densification operator}

Constraint languages for which the densification problem ${\rm Dense}(\boldsymbol{\Gamma})$ is tractable can be classified using tools of universal algebra. An analogous approach can be applied to classify languages with a weak polynomial densification operator. 

\begin{definition} Let $\boldsymbol{\Gamma} = (D,\varrho_1, ..., \varrho_s)$ and $\tau =\{\pi_1, ..., \pi_s\}$. A $k$-ary relation $\rho\in \langle \Gamma\rangle$ is called strongly reducible to $\Gamma$ if there exists a quantifier-free primitive positive formula $\Xi(x_1,\cdots, x_n)$\footnote{A quantifier-free pp-formula is a pp-formula without existential quantification.} (over $\tau$) and $\delta\subseteq D^n$ for some $n\geq k$ 
such that ${\rm pr}_{1:k}\Xi^{\boldsymbol{\Gamma}}=\rho$, ${\rm pr}_{1:k}\delta = D^k\setminus \rho$ and $\Xi^{\boldsymbol{\Gamma}}\cup \delta\in \langle \Gamma\rangle$. A $k$-ary relation $\rho\in \langle \Gamma\rangle$ is called A-reducible to $\Gamma$ if $\rho = \rho_1\cap \cdots \cap \rho_l$, where $\rho_i\in \langle \Gamma\rangle$ is strongly reducible to $\Gamma$ for $i\in [l]$.
\end{definition}

\begin{definition} A constraint language $\Gamma$ is called an A-language if any $\rho\in \langle \Gamma\rangle$ is A-reducible to $\Gamma$. 
\end{definition}
Examples of A-languages are stated in the following theorems, whose proofs can be found in Section~\ref{Horn-case-reduce-proof}.
\begin{theorem}\label{SAT-case-reduce} Let $\boldsymbol{\Gamma} = (D=\{0,1\}, \varrho_1, \varrho_2, \varrho_3)$ where $ \varrho_1 = \big \{ ({x, y} ) | x \vee y \big \} $,
 $ \varrho_2 = \big \{ ({x, y} ) | \neg x \vee y \big \} $
 and $ \varrho_3 = \big \{{ ({x, y} ) | \neg x \vee \neg y} \big \} $. Then, $\Gamma$ is an  A-language.
\end{theorem}
\begin{theorem}\label{Horn-case-reduce} Let $\boldsymbol{\Gamma} = (D=\{0,1\}, \{(0)\}, \{(1)\}, \varrho_{x \wedge y\rightarrow z})$ where $\varrho_{x \wedge y\rightarrow z} = \{(a_1, a_2, a_{3})\in D^{3}| a_1 a_{2}\leq a_{3}\}$. Then, $\Gamma$ is an  A-language.
\end{theorem}

Reducibility of a relation to a language is an interesting notion because of its property stated in the following theorem.

\begin{theorem}\label{reduction} Let $\Gamma, \Gamma'$ be constraint languages such that $\Gamma'\subseteq \langle\Gamma\rangle$, and every relation in $\Gamma'$ is A-reducible to $\Gamma$. Then: 
\begin{itemize}
\item[{\em (a)}] ${\rm Dense}(\boldsymbol{\Gamma}')$ is polynomial-time Turing reducible to  ${\rm Dense}(\boldsymbol{\Gamma})$;
\item[{\em (b)}] if $\boldsymbol{\Gamma}$ has a weak polynomial densification operator, then $\boldsymbol{\Gamma}'$ also has a weak polynomial densification operator;
\item[{\em (c)}] if ${\rm Dense}(\boldsymbol{\Gamma})\in {\rm mP/poly}$, then ${\rm Dense}(\boldsymbol{\Gamma}')\in {\rm mP/poly}$.
\end{itemize}
\end{theorem}
\begin{proof} 
Since $\Gamma'\subseteq \langle\Gamma\rangle$, then there is $L = \{\Phi_i | \ i\in [c]\}$ where $\Phi_i$ is a primitive positive formula over the vocabulary $\tau=\{\pi_1, ..., \pi_s\}$, such that $\boldsymbol{\Gamma} = (D,\varrho_1, ..., \varrho_s)$, $\boldsymbol{\Gamma}' = (D,\Phi^{\boldsymbol{\Gamma}}_1, ..., \Phi^{\boldsymbol{\Gamma}}_c)$.

Let $\mathbf{R}' = (V,r'_1, ..., r'_c)$ be an instance of ${\rm Dense}(\boldsymbol{\Gamma}')$. Our goal is to compute a maximal instance $(\mathbf{R}''= (V,r''_1, ..., r''_c), \boldsymbol{\Gamma}')$ such that $r''_i\supseteq r'_i, i\in [c]$ and ${\rm Hom}(\mathbf{R}'', \boldsymbol{\Gamma}') = {\rm Hom}(\mathbf{R}', \boldsymbol{\Gamma}')$, or in other words, to compute ${\rm Dense}(\mathcal{C}_{\mathbf{R}' })$.  

First, let us introduce some notations. Let $\Psi$ be any primitive positive formula over $\tau$, i.e. $\Psi = \exists x_{k+1}... x_{l} \bigwedge_{t\in [N]}\pi_{j_{t}} (x_{o_{t1}}, x_{o_{t2}}, ...)$ where $j_{t}\in [s]$ and $o_{tx}\in [l]$ and $\mathbf{a} = (a_1, ..., a_{k})$ be a tuple of objects.
Let us introduce a set of new distinct objects ${\rm NEW}(\mathbf{a},\Psi)=\{a_{k+1}, ..., a_{l}\}$. Note that the sets ${\rm NEW}(\mathbf{a}, \Psi)$ are disjoint for different $(\mathbf{a}, \Psi)$ (also, ${\rm NEW}(\mathbf{a}, \Psi)\cap V = \emptyset$). 
For a tuple $\mathbf{a} = (a_1, ..., a_{k})$, the constraint that an assignment to $(a_1, ..., a_{k})$ is in $\Psi^{\boldsymbol{\Gamma}}$ can be expressed by a collection of constraints $\mathfrak{C}({\mathbf a}, \Psi) = \{\langle (a_{o_{t1}}, a_{o_{t2}}, ...),  \varrho_{j_{t}}\rangle \mid t\in [N]\}$. In other words, we require that an image of $(a_{o_{t1}}, a_{o_{t2}}, ...)$ is in $\varrho_{j_{t}}$ for any $t\in [N]$.  Note that $\mathfrak{C}({\mathbf a}, \Psi) $ is a set of constraints over a set of variables $\{a_1, ..., a_{k}\}\cup {\rm NEW}(\Psi,\mathbf{a})$ where only relations from $\Gamma$ are allowed. 

Let us start with a proof of statement (a).  We will describe a reduction to ${\rm Dense}(\boldsymbol{\Gamma})$ that consists of two steps: first we add new variables and construct an instance of ${\rm CSP}(\boldsymbol{\Gamma})$ in the same way as it is done in the standard reduction of ${\rm CSP}(\boldsymbol{\Gamma}')$ to ${\rm CSP}(\boldsymbol{\Gamma})$; afterwards, we add new variables and constraints and form an instance of ${\rm Dense}(\boldsymbol{\Gamma})$.

First, for any $i\in [c],\mathbf{a}\in r'_i$, we add objects ${\rm NEW}(\mathbf{a}, \Phi_i)$ to the set of variables $V$ and define an extended set $M^0 = V\cup \bigcup_{i\in [c],\mathbf{a}\in r'_i}{\rm NEW}(\mathbf{a}, \Phi_i)$. Afterwards, we define a relational structure $(\mathbf{R}^0= (M^0,r^0_1, ..., r^0_s), \boldsymbol{\Gamma})$ where $\mathcal{C}_{\mathbf{R}^0} = \bigcup_{i\in [c],\mathbf{a}\in r'_i}\mathfrak{C}({\mathbf a}, \Phi_i)$. By construction, ${\rm pr}_{V}{\rm Hom}(\mathbf{R}^0, \boldsymbol{\Gamma}) = {\rm Hom}(\mathbf{R}', \boldsymbol{\Gamma}')$. Note that this reduction is standard in the algebraic approach to fixed-template CSPs. This is the first step of the construction.

Let us now consider a relation $\Phi_i^{\boldsymbol{\Gamma}}$ and assume that its arity is $k$.
According to the assumption, $\Phi_i^{\boldsymbol{\Gamma}}$ is A-reducible to $\Gamma$. Therefore, $\Phi_i^{\boldsymbol{\Gamma}} = \varrho_{i1}\cap \cdots \cap \varrho_{il}$, where $\varrho_{ij}$ is strongly reducible to $\Gamma$ for $j\in [l]$. Thus, there exists a quantifier-free primitive positive formula over $\tau$, $\Xi_j$, involving $r_j$ variables, and $\delta_j\subseteq D^{r_j}$, such that $\varrho_{ij} = {\rm pr}_{1:k}\Xi_j^{\boldsymbol{\Gamma}}$ and ${\rm pr}_{1:k}\delta_j = D^{k}\setminus \varrho_{ij}$ and $\delta_j\cup \Xi_j^{\boldsymbol{\Gamma}}\in \langle \Gamma\rangle$. Since $\gamma_j = \delta_j\cup \Xi_j^{\boldsymbol{\Gamma}}$ is pp-definable over $\Gamma$, there exists a primitive positive formula over $\tau$, $\exists x_{r_j+1}\cdots x_{p_j}\Theta_j (x_1, \cdots, x_{p_j})$ where $\Theta_j$ is quantifier-free, such that $(\exists x_{r_j+1}\cdots x_{p_j}\Theta_j (x_1, \cdots, x_{p_j}))^{\boldsymbol{\Gamma}} = \delta_j\cup \Xi_j^{\boldsymbol{\Gamma}}$. Let us introduce a set of constraints:
$$
\mathfrak{C}(V, \Phi_i) = \bigcup_{(a_1, ..., a_{k})\in V^k} \bigcup_{j\in [l]} \mathfrak{C}\big((a_1, ..., a_{k}), \exists x_{k+1}, \cdots,x_{p_j}\Theta_j (x_1, \cdots, x_{p_j})\big).
$$
over a set of variables $$M_i = V\cup \bigcup_{(a_1, ..., a_{k})\in V^k} \bigcup_{j\in [l]} {\rm NEW}\big((a_1, ..., a_{k}), \exists x_{k+1}, \cdots,x_{p_j}\Theta_j (x_1, \cdots, x_{p_j})).$$
Due to ${\rm pr}_{1:k}\delta_j = D^{k}\setminus \varrho_{ij}$, we have ${\rm pr}_{1:k} (\delta_j\cup \Xi_j^{\boldsymbol{\Gamma}}) = D^k$. Therefore, $$(\exists x_{k+1}\cdots x_{p_j}\, \Theta_j (x_1, \cdots, x_{p_j}))^{\boldsymbol{\Gamma}} = {\rm pr}_{1:k}(\delta_j\cup \Xi_j^{\boldsymbol{\Gamma}}) = D^k.$$
Thus, the set of constraints $\mathfrak{C}(V, \Phi_i)$ does not add any restrictions on assignments of $V$ (though it adds restrictions on additional variables).

Let $\mathbf{R}= (M,r_1, ..., r_s)$ be such that $M = V \cup \bigcup_{i\in [c],\mathbf{a}\in r'_i}{\rm NEW}({\mathbf a}, \Phi_i) \bigcup_{i\in [c]} M_i$ and $\mathcal{C}_{\mathbf{R}} = \bigcup_{i\in [c],\mathbf{a}\in r'_i}\mathfrak{C}({\mathbf a}, \Phi_i) \bigcup_{i\in [c]} \mathfrak{C}(V, \Phi_i)$. By construction, ${\rm pr}_{V}{\rm Hom}(\mathbf{R}, \boldsymbol{\Gamma}) = {\rm Hom}(\mathbf{R}', \boldsymbol{\Gamma}')$. Let us treat $\mathbf{R}$ as an instance of ${\rm Dense}(\boldsymbol{\Gamma})$.

The computation of ${\rm Dense}(\mathcal{C}_{\mathbf{R}'})$ can be made by checking whether $\langle (v_1,\cdots, v_k), \Phi^{\boldsymbol{\Gamma}}_i \rangle\in {\rm Dense}(\mathcal{C}_{\mathbf{R}'})$ for any $v_1,\cdots, v_k\in V$ and a $k$-ary $\Phi^{\boldsymbol{\Gamma}}_i\in \Gamma$. 
From the following lemma it follows that such a checking can be reduced to a checking of certain conditions  of the form $\langle (u_{1}, u_{2}, ...), \varrho_{j}\rangle \in {\rm Dense}(\mathcal{C}_{\mathbf{R}})$, i.e. to the computation of $ {\rm Dense}(\mathcal{C}_{\mathbf{R}})$. 
\begin{lemma} For a $k$-ary $\Phi^{\boldsymbol{\Gamma}}_i$ and $v_1,\cdots, v_k\in V$ there is a subset $S_i(v_1,\cdots, v_k) \subseteq \mathcal{C}^{\boldsymbol{\Gamma}}_M$ (that can be computed in time ${\rm poly}(|V|)$) such that the condition $\langle (v_1,\cdots, v_k), \Phi^{\boldsymbol{\Gamma}}_i \rangle\in {\rm Dense}(\mathcal{C}_{\mathbf{R}'}) \,(\subseteq \mathcal{C}^{\boldsymbol{\Gamma}'}_V)$  is equivalent to a list of conditions $\langle (u_{1}, u_{2}, ...), \varrho_{j}\rangle \in {\rm Dense}(\mathcal{C}_{\mathbf{R}})\, (\subseteq \mathcal{C}^{\boldsymbol{\Gamma}}_M)$ for $\langle (u_{1}, u_{2}, ...), \varrho_{j}\rangle  \in S_i(v_1,\cdots, v_k)$.
\end{lemma}
\begin{proof}
Note that $\langle (v_1,\cdots, v_k), \Phi^{\boldsymbol{\Gamma}}_i \rangle\in {\rm Dense}(\mathcal{C}_{\mathbf{R}'})\subseteq \mathcal{C}^{\boldsymbol{\Gamma}'}_V$ for $v_1,\cdots, v_k\in V$ if and only if ${\rm pr}_{v_1,\cdots, v_k} {\rm Hom}(\mathbf{R}, \boldsymbol{\Gamma})\subseteq \Phi^{\boldsymbol{\Gamma}}_i$. Let us assume that we have ${\rm pr}_{v_1,\cdots, v_k} {\rm Hom}(\mathbf{R}, \boldsymbol{\Gamma})\subseteq \Phi^{\boldsymbol{\Gamma}}_i$. The definition of $\mathbf{R}$ implies that we have a set of constraints $$\mathfrak{C} ((v_1, ..., v_{k}), \exists x_{k+1}, \cdots,x_{p_j}\Theta_j (x_1, \cdots, x_{p_j}))$$ imposed on $v_1,\cdots, v_k$ and $${\rm NEW}\big((v_1, ..., v_{k}), \exists x_{k+1}, \cdots,x_{p_j}\Theta_j (x_1, \cdots, x_{p_j}))  = \{v_{k+1},\cdots, v_{p_j}\}$$ (how $\Phi_i$ and $\Theta_{j}, j\in [l]$ are related is described above). Since $\Phi_i^{\boldsymbol{\Gamma}} = \varrho_{i1}\cap \cdots \cap \varrho_{il}$, we conclude ${\rm pr}_{v_1,\cdots, v_k} {\rm Hom}(\mathbf{R}, \boldsymbol{\Gamma})\subseteq \varrho_{ij}, j\in [l]$. Therefore, ${\rm pr}_{v_1,\cdots, v_{p_j}} {\rm Hom}(\mathbf{R}, \boldsymbol{\Gamma})\subseteq \{{\mathbf x}\in \Theta_j^{ \boldsymbol{\Gamma}} \mid {\mathbf x}_{1:k}\in \varrho_{ij}\}$, that is ${\rm pr}_{v_1,\cdots, v_{r_j}} {\rm Hom}(\mathbf{R}, \boldsymbol{\Gamma})\subseteq \{{\mathbf x}_{1:r_j}\mid {\mathbf x}\in \Theta_j^{ \boldsymbol{\Gamma}}, {\mathbf x}_{1:k}\in \varrho_{ij}\} = \Xi_j^{ \boldsymbol{\Gamma}}$. Since $\Xi_j$ is a quantifier-free primitive positive formula over $\tau$, then the fact ${\rm pr}_{v_1,\cdots, v_{r_j}} {\rm Hom}(\mathbf{R}, \boldsymbol{\Gamma})\subseteq \Xi_j^{ \boldsymbol{\Gamma}}$ can be expressed as $(h(v_1),\cdots, h(v_{r_j}))\in \Xi^{\boldsymbol{\Gamma}}_j$ for any $h\in {\rm Hom}(\mathbf{R}, \boldsymbol{\Gamma})$. In other words, if  $\Xi_j = \exists x_{k+1}... x_{l} \bigwedge_{t\in [N]}\pi_{w_{t}} (x_{o_{t1}}, x_{o_{t2}}, ...)$, then $\langle (v_{o_{t1}}, v_{o_{t2}}, ...), \varrho_{w_t}\rangle \in {\rm Dense}(\mathcal{C}_{\mathbf{R}})\subseteq \mathcal{C}^{\boldsymbol{\Gamma}}_V$ for any $t\in [N]$. Let us set
$$
S_i(v_1,\cdots, v_k) = \{\langle (v_{o_{t1}}, v_{o_{t2}}, ...), \varrho_{w_t}\rangle \mid \Xi_j = \exists x_{k+1}... x_{l} \bigwedge_{t\in [N]}\pi_{w_{t}} (x_{o_{t1}}, x_{o_{t2}}, ...), j\in [l]\}
$$
In fact, we proved $$
\langle (v_1,\cdots, v_k), \Phi^{\boldsymbol{\Gamma}}_i \rangle\in {\rm Dense}(\mathcal{C}_{\mathbf{R}'}) \Rightarrow S_i(v_1,\cdots, v_k)\subseteq {\rm Dense}(\mathcal{C}_{\mathbf{R}}).$$
It can be easily checked that the last chain of arguments can be reversed, and   
$$S_i(v_1,\cdots, v_k)\subseteq {\rm Dense}(\mathcal{C}_{\mathbf{R}})\Rightarrow \langle (v_1,\cdots, v_k), \Phi^{\boldsymbol{\Gamma}}_i \rangle\in {\rm Dense}(\mathcal{C}_{\mathbf{R}'}).
$$
\end{proof}

Thus, statement (a) is proved.

Statement (b) directly follows from the previous reduction.
Suppose $\boldsymbol{\Gamma}$ has a weak polynomial densification operator, i.e. there is a finite $S_n\supseteq \mathcal{C}_n^{\boldsymbol{\Gamma}}$ and an implicational system $\Delta_n \subseteq 2^{S_n}\times 2^{S_n}$ of size $|\Delta_n| = {\mathcal O}({\rm poly}(n))$ that acts on ${\mathcal C}^{\boldsymbol{\Gamma}}_{n}$ as the densification operator, i.e.  $\Sigma_n^{\boldsymbol{\Gamma}} = \{(A\rightarrow B)\in \Delta_n^\triangleright| A,B\subseteq {\mathcal C}^{\boldsymbol{\Gamma}}_{n}\}$.

If $V = [n]$, then 
$X = V \cup \bigcup_{i\in [c],\mathbf{a}=(a_1,a_2, \cdots,), a_i\in V}{\rm NEW}({\mathbf a}, \Phi_i) \bigcup_{i\in [c]} M_i$ ($M_i$ are defined above) 
is a superset of $V$ whose size is bounded by a polynomial of $n$. Therefore, w.l.o.g. we can assume $X = [m]$ where $m=|X|={\mathcal O}({\rm poly} (n))$. 
Let $\Delta_m$ be an implicational system on $S_m\supseteq {\mathcal C}^{\boldsymbol{\Gamma}}_{m}$ such that $|\Delta_m| = {\mathcal O}({\rm poly} (m))$ and $o_{\Delta_m}(S) = \{x\in {\mathcal C}^{\boldsymbol{\Gamma}}_{m} | (S\rightarrow x)\in  \Delta_m^\triangleright\}$ acts as the densification operator on subsets of ${\mathcal C}^{\boldsymbol{\Gamma}}_{m}$. 
Since $\Delta_m\subseteq 2^{S_m} \times 2^{S_m}$,  we can interpret $\Delta_m$ as an implicational system on $S'_m = S_m\cup {\mathcal C}^{\boldsymbol{\Gamma'}}_{n}$, i.e. we include ${\mathcal C}^{\boldsymbol{\Gamma'}}_{n}$ into a set of literals of $\Delta_m$.
Let us now add to $\Delta_m$ new implications by the following rule: for $\Phi_i = \exists x_{k+1}... x_{l} \bigwedge_{t\in [N]}\pi_{j_{t}} (x_{o_{t1}}, x_{o_{t2}}, ...)$, $\mathbf{a}\in [n]^{k}$ and the corresponding new $l-k$ variables ${\rm NEW}(\mathbf{a}, \Phi_i) = \{a_{k+1}, ..., a_{l}\}$ we add $R(\mathbf{a}, \Phi_i): \langle \mathbf{a}, \Phi_i^{\boldsymbol{\Gamma}} \rangle\rightarrow \{\langle(a_{o_{t1}}, a_{o_{t2}}, ...), \varrho_{j_{t}} \rangle | t\in [N]\}$. Let us denote 
$$
\mathfrak{R}_1 = \bigcup_{i\in [c], \mathbf{a} = (a_1, a_2,...): a_i\in V} \{R(\mathbf{a}, \Phi_i)\}.
$$
The second kind of implications that we need to add to $\Delta_m$ is
$$
\mathfrak{R}_2 = \bigcup_{i\in [c]}\{\emptyset \to \mathfrak{C}(V, \Phi_i) \}.
$$
The last set of implications, $\mathfrak{R}_3$, is defined by
$$\mathfrak{R}_3 = \{(S_i(v_1,\cdots, v_k) \to \langle (v_1,\cdots, v_k), \Phi^{\boldsymbol{\Gamma}}_i\rangle )\mid \langle (v_1,\cdots, v_k), \Phi^{\boldsymbol{\Gamma}}_i\rangle \in \mathcal{C}^{\boldsymbol{\Gamma}'}_{n}\},$$ where $S_i(v_1,\cdots, v_k)$ is described in the previous Lemma, i.e. it equals a set of constraints for which $S_i(v_1,\cdots, v_k)\subseteq {\rm Dense}(\mathcal{C}_{\mathbf {R}})$ is equivalent to $\langle (v_1,\cdots, v_k), \Phi^{\boldsymbol{\Gamma}}_i\rangle\in {\rm Dense}(\mathcal{C}_{\mathbf {R}'})$. Thus, we defined a set of implications $\Delta_m\cup \mathfrak{R}_1\cup \mathfrak{R}_2 \cup \mathfrak{R}_3$.
Let us denote a new system by $\Sigma_n$. By the construction of $\Sigma_n$, we have $|\Sigma_n| = {\mathcal O}({\rm poly} (n))$. 

Given $\mathcal{C}_{\mathbf {R}'}$, using implications from $\mathfrak{R}_1$, one can derive the set of constraints $\mathcal{C}_{\mathbf {R}^0}$ ($\mathbf {R}^0$ is defined above), and using  implications from $\mathfrak{R}_2$ one completes the set of derivable literals to $\mathcal{C}_{\mathbf {R}}$. Then, using initial rules of $\Delta_m$, one can derive from $\mathcal{C}_{\mathbf {R}}$ its closure ${\rm Dense}(\mathcal{C}_{\mathbf {R}})$. Finally, using implications from $\mathfrak{R}_3$ one can derive all constraints from ${\rm Dense}(\mathcal{C}_{\mathbf {R}'})$. It is not hard to prove that  $x\in \mathcal{C}_n^{\boldsymbol{\Gamma}'}$ is derivable from $\mathcal{C}_{\mathbf {R}'}$ if and only if $x\in{\rm Dense}(\mathcal{C}_{\mathbf {R}'})$.

Thus, $\boldsymbol{\Gamma}'$ also has a weak polynomial densification operator. 
Note that implications $\mathfrak{R}_2 \cup \mathfrak{R}_3$ are all from $\Sigma^{\boldsymbol{\Gamma\cup \Gamma'}}_{m}$, but an implication $R(\mathbf{a}, \Phi_i)\in \mathfrak{R}_1$ is not, in general, from $\Sigma^{\boldsymbol{\Gamma\cup \Gamma'}}_{m}$. 

Statement (c) directly follows from the fact that the function $Q: 2^{\mathcal{C}^{{\mathbf \Gamma}'}_V}\to 2^{\mathcal{C}^{\mathbf \Gamma}_M}$ such that $Q(C_{{\mathbf R}'})=C_{\mathbf R}$ is monotone and can be computed by a polynomial-size monotone circuit.
\end{proof}

\section{DS-basis and algorithms for ${\rm Dense}(\boldsymbol{\Gamma})$ and ${\rm Sparse}(\boldsymbol{\Gamma})$}\label{DS-section}
The notion of DS-basis is a formalization of templates for which a small cover of $\Sigma_n^{\boldsymbol{\Gamma}}$ not only exists but can also  be computed efficiently.  
\begin{definition} A fixed template $\boldsymbol{\Gamma}$ is called a DS-basis, if there exists an algorithm $\mathcal{A}$ that solves in time ${\mathcal O}({\rm poly}(n))$ the task with:
\begin{itemize}
\item An instance: a natural number $n\in {\mathbb N}$;
\item An output: an implicational system $\Sigma\subseteq \Sigma_{n}^{\boldsymbol{\Gamma}}$ such that $\Sigma^\triangleright= \Sigma_n^{\boldsymbol{\Gamma}}$.
\end{itemize}
\end{definition}

\begin{theorem}\label{dense-algebra}
For any DS-basis $\boldsymbol{\Gamma}$ there is an algorithm $\mathcal{A}_1$ that, given an instance $\mathbf{R}$ of ${\rm Dense}(\boldsymbol{\Gamma})$, solves the densification problem for $(\mathbf{R}, \boldsymbol{\Gamma})$ in time ${\mathcal O}({\rm poly}(|V|))$.
\end{theorem}
\begin{proof}
For any implicational system $\Sigma\subseteq 2^S \times 2^S$, and any $A,B\subseteq S$, the membership $A\rightarrow B \mathop\in\limits^{?} \Sigma^\triangleright$ can be checked in time ${\mathcal O}(|\Sigma|)$ by Beeri and Bernstein's algorithm for functional dependencies~\cite{Beeri}.

Since $\boldsymbol{\Gamma}$ is the DS-basis, then there exists an algorithm $\mathcal{A}$ using which one can compute in time ${\mathcal O}({\rm poly}(|V|))$ an implicational system $\Sigma\subseteq \Sigma_V^{\boldsymbol{\Gamma}}$ such that $\Sigma^\triangleright = \Sigma_V^{\boldsymbol{\Gamma}}$. Afterwards, we check whether $\mathcal{C}_\mathbf{R}\rightarrow x\mathop\in\limits^{?} \Sigma_V^{\boldsymbol{\Gamma}}$ using Beeri and Bernstein's algorithm for any $x\in \mathcal{C}^{\boldsymbol{\Gamma}}_V$ and compute ${\rm Dense} (\mathcal{C}_\mathbf{R}) = \{x\in \mathcal{C}^{\boldsymbol{\Gamma}}_V| \mathcal{C}_\mathbf{R}\rightarrow x\in \Sigma^\triangleright\}$ in time ${\mathcal O}(|\mathcal{C}^{\boldsymbol{\Gamma}}_V|\cdot |\Sigma|) = {\mathcal O}({\rm poly}(|V|))$. Finally we set $r'_i = \{(v_1, ..., v_{\vectornorm{\varrho_i}}) | \langle (v_1, ..., v_{\vectornorm{\varrho_i}}), \varrho_i\rangle \in {\rm Dense} (\mathcal{C}_\mathbf{R})\}$ for $i\in [s]$. The instance $(\mathbf{R}'= (V,r'_1, ..., r'_s), \boldsymbol{\Gamma})$ is maximal.
\end{proof}
The following theorem is equivalent to Theorem~\ref{main-sparse} announced in Section~\ref{main-res}. 
\begin{theorem}\label{sparse-algebra}
For any DS-basis $\boldsymbol{\Gamma}$ there is an algorithm $\mathcal{A}_2$ that, given an instance $\mathbf{R}$ of ${\rm Sparse}(\boldsymbol{\Gamma})$, solves the sparsification problem for $(\mathbf{R}, \boldsymbol{\Gamma})$ in time ${\mathcal O} ({\rm poly}(|V|)\cdot |{\rm Min} (\mathbf{R}, \boldsymbol{\Gamma})|^2)$.
\end{theorem}
\begin{proof}
It is easy to see that a set of all possible instances of  ${\rm Sparse}(\boldsymbol{\Gamma})$, $\{\mathbf{R} = (V, \cdots)\}$, is in one-to-one correspondence with a set $2^{\mathcal{C}^{\boldsymbol{\Gamma}}_V}$.
For any implicational system $F$ on $S$, 
let us call $A\subseteq S$ a minimal key of $F$ for $B$ if $(A\rightarrow B)\in F^{\triangleright}$, but for any proper subset $C\subset A$, $(C\rightarrow B)\notin F^{\triangleright}$. Let us prove first that $\mathbf{R}'\in {\rm Min} (\mathbf{R}, \boldsymbol{\Gamma})$ is and only if $\mathcal{C}_{\mathbf{R}'}$ is a minimal key of $\Sigma^{\boldsymbol{\Gamma}}_V$ for ${\rm Dense} (\mathcal{C}_\mathbf{R})$.

Indeed, if $\mathbf{R}'\in {\rm Min} (\mathbf{R}, \boldsymbol{\Gamma})$, then ${\rm Hom}(\mathbf{R}, \boldsymbol{\Gamma})= {\rm Hom}(\mathbf{R}', \boldsymbol{\Gamma})$. Since ${\rm Hom}(\mathbf{R}, \boldsymbol{\Gamma})= {\rm Hom}(\mathbf{R}', \boldsymbol{\Gamma})$, then ${\rm Dense} (\mathcal{C}_{\mathbf{R}}) = {\rm Dense} (\mathcal{C}_{\mathbf{R}'})$ (by the definition of the densification operator). Therefore, from the duality between the closure operator ${\rm Dense}$ and the implication system $\Sigma_V^{\boldsymbol{\Gamma}}$ we obtain $(\mathcal{C}_{\mathbf{R}'} \rightarrow {\rm Dense} (\mathcal{C}_\mathbf{R}))\in \Sigma_V^{\boldsymbol{\Gamma}}$. Since the pair $(\mathbf{R}', \boldsymbol{\Gamma})$ is minimal, we obtain that $\mathcal{C}_{\mathbf{R}'}$ is a minimal key for ${\rm Dense} (\mathcal{C}_\mathbf{R})$.

On the contrary, let $\mathcal{C}_{\mathbf{R}'}$ be a minimal key for ${\rm Dense} (\mathcal{C}_\mathbf{R})$. Therefore, ${\rm Dense} (\mathcal{C}_\mathbf{R}) = {\rm Dense} (\mathcal{C}_{\mathbf{R}'})$,  from which we obtain ${\rm Hom}(\mathbf{R}, \boldsymbol{\Gamma})= {\rm Hom}(\mathbf{R}', \boldsymbol{\Gamma})$. Any proper subset $\mathcal{C}_{\mathbf{R}''}\subset \mathcal{C}_{\mathbf{R}'}$ has a closure ${\rm Dense} (\mathcal{C}_{\mathbf{R}''})\subset {\rm Dense} (\mathcal{C}_{\mathbf{R}'})$. Thus, we obtain that ${\rm Hom}(\mathbf{R}', \boldsymbol{\Gamma})\ne {\rm Hom}(\mathbf{R}'', \boldsymbol{\Gamma})$ (otherwise, we have ${\rm Dense} (\mathcal{C}_{\mathbf{R}''}) = {\rm Dense} (\mathcal{C}_{\mathbf{R}'})$).
We conclude that the pair $(\mathbf{R}', \boldsymbol{\Gamma})$ is minimal.

Since $\boldsymbol{\Gamma}$ is a DS-basis, we construct in advance an implicational system $\Sigma\subseteq\Sigma_V^{\boldsymbol{\Gamma}}$ such that $\Sigma^\triangleright = \Sigma_V^{\boldsymbol{\Gamma}}$.
We proved that the problem of listing of ${\rm Min} (\mathbf{R}, \boldsymbol{\Gamma})$ is equivalent to a listing of all minimal keys for ${\rm Dense} (\mathcal{C}_\mathbf{R})$ in the implicational system $\Sigma$. In database theory, this task is called the optimal cover problem and was studied in the 70s~\cite{CODD}. The algorithm of Luchessi and Osborn lists all minimal keys for ${\rm Dense} (\mathcal{C}_\mathbf{R})$ in time $\mathcal{O}(|\Sigma|\cdot |{\rm Min} (\mathbf{R}, \boldsymbol{\Gamma})| \cdot |{\rm Dense} (\mathcal{C}_\mathbf{R})| \cdot (|{\rm Min} (\mathbf{R}, \boldsymbol{\Gamma})|+|{\rm Dense} (\mathcal{C}_\mathbf{R})|))$ (see p. 274 of~\cite{LUCCHESI1978270}). It is easy to see that the last expression is bounded by ${\mathcal O} ({\rm poly}(|V|)\cdot |{\rm Min} (\mathbf{R}, \boldsymbol{\Gamma})|^2)$.

Note that main approaches to listing minimal keys in a functional dependency table refer to the method of Luchessi and Osborn. Nowadays, several alternative methods are designed for this and adjacent tasks~\cite{BenitoPicazo2017}, including efficient parallelization techniques~\cite{Sridhar}.
\end{proof}
\begin{remark}\label{Gottlob} Sometimes we are interested not in ${\rm Min}(\mathbf{R}, \boldsymbol{\Gamma})$, but in its subset ${\rm Min}(\mathbf{R}, \boldsymbol{\Gamma}, S) = \{\mathbf{R}'\in {\rm Min}(\mathbf{R}, \boldsymbol{\Gamma}) \mid \mathcal{C}_{\mathbf{R}'}\subseteq S\}$ where $S\subseteq \mathcal{C}_{V}^{\boldsymbol{\Gamma}}$. For example, if $S=\mathcal{C}_{\mathbf{R}}$, then listing ${\rm Min}(\mathbf{R}, \boldsymbol{\Gamma}, S)$ is equivalent to a listing of all non-redundant sparsifications that are subsets of the set of initial constraints. The latter set could have a substantially smaller cardinality than ${\rm Min}(\mathbf{R}, \boldsymbol{\Gamma})$. A natural approach to list  ${\rm Min}(\mathbf{R}, \boldsymbol{\Gamma}, S)$ is to compute a cover $\Sigma'$ of $\Sigma_V^{\boldsymbol{\Gamma}}\cap (2^{S})^2= \Sigma^\triangleright\cap (2^{S})^2$ and then  list minimal keys of $\Sigma'$ for $S$ (sometimes called candidate keys) by the method of Luchessi and Osborn in time $\mathcal{O}(|\Sigma'|\cdot |{\rm Min} (\mathbf{R}, \boldsymbol{\Gamma}, S)| \cdot |S| \cdot (|{\rm Min} (\mathbf{R}, \boldsymbol{\Gamma}, S)|+|S|))$. For  the computation of $\Sigma'$, it is natural to exploit the Reduction by Resolution algorithm (RBR)  suggested in~\cite{Gottlob}. The bottleneck of that strategy is that a small cover of $\Sigma^\triangleright\cap(2^{\mathcal{C}_{\mathbf{R}}})^2$ may not exist. In such cases RBR's computation takes a long time that can be  potentially exponential.
\end{remark}

Next, we will show that DS-bases include such templates for which ${\rm Dense} (\boldsymbol{\Gamma})$ can be solved by a Datalog program.

\section{Densification by Datalog program}
The idea of using Datalog programs for CSP is classical~\cite{Feder,BODIRSKY201379,Larose}.

\begin{definition} 
If $\Phi(x_1, ..., x_{n_u})$ is a primitive positive formula over $\tau$, then the first-order formula $$\Psi = \forall x_1, ..., x_{n_u} \big(\Phi (x_1, ..., x_{n_u})\rightarrow \pi_u (x_1, ..., x_{n_u})\big)$$ is called a Horn formula\footnote{We slightly abuse the standard terminology, according to which Horn formulas are defined more generally.} over $\tau$. If a primitive positive definition of $\Phi$ involves $n$ variables, then $\Psi$ is said to be of width $(n_u, n)$ (or, simply, of width $n$).
Any Horn formula of width $(n_u, n)$ is equivalent to the universal formula $$\forall x_1, ..., x_{n} \big( \bigwedge_{t=1}^N \pi_{j_t} (x_{o_{t1}}, x_{o_{t2}}, ..., x_{o_{tn_{j_t}}})\rightarrow \pi_u (x_1, ..., x_{n_u})\big),$$ so we will refer to both of them as Horn formulas.  For a relational structure $\mathbf{R} = (V,r_1, ..., r_s)$, $\vectornorm{r_i} = n_i$, $\mathbf{R}\vDash\Psi $ denotes $\Phi^{\mathbf{R}}\subseteq r_{u}$.
\end{definition}

For the densification task, the use of Datalog is motivated by the following theorem.

\begin{theorem}\label{closure} Let $(\mathbf{R}, \boldsymbol{\Gamma})$ be a maximal instance of CSP. For any Horn formula $\Psi$, if $\boldsymbol{\Gamma}\vDash\Psi$, then $\mathbf{R}\vDash\Psi $.
\end{theorem}
\begin{proof}
Let $\boldsymbol{\Gamma} = (D, \varrho_1, ..., \varrho_s)$ and $$\Psi = \forall x_1, ..., x_{n_u} \exists x_{n_u+1}... x_{n} \Xi(x_1, ..., x_n)\rightarrow \pi_u(x_1, ..., x_{n_u})$$ where $$\Xi(x_1, ..., x_n) = \bigwedge_{t=1}^N \pi_{j_t} (x_{o_{t1}}, x_{o_{t2}}, ..., x_{o_{tn_{j_t}}})$$ such that $\boldsymbol{\Gamma}\vDash\Psi$. Let $h: V\rightarrow D$ be any mapping and $r_i = h^{-1}(\varrho_i)$. Let us prove that $\mathbf{R}\vDash\Psi$ where $\mathbf{R} = (V,r_1, ..., r_s)$. 

Indeed, for any ${\mathbf a}\in r_i$ we have $h({\mathbf a})\in \varrho_i$, $i\in [s]$. From $\boldsymbol{\Gamma}\vDash\Psi$ we obtain that the following statement is true: if there exist $a_{1},..., a_{n}\in D$ such that $ (a_{o_{t1}}, a_{o_{t2}}, ..., a_{o_{tn_{j_t}}})\in \varrho_{j_t}$, $t\in [N]$, then $(a_1, ..., a_{n_u})\in \varrho_u$.

Suppose now that we are given $b_{1}, ..., b_{n}\in V$ such that for any $t\in [N]$ we have $ (b_{o_{t1}}, b_{o_{t2}}, ..., b_{o_{tn_{j_t}}})\in r_{j_t}$. Therefore, for any $t\in [N]$ we  have $$ (h(b_{o_{t1}}), h(b_{o_{t2}}), ..., h(b_{o_{tn_{j_t}}}))\in \varrho_{j_t}.$$ From $\boldsymbol{\Gamma}\vDash\Psi$ we obtain that $(h(b_1), ..., h(b_{n_u}))\in \varrho_u$. Therefore, $(b_1, ..., b_{n_u})\in r_u$. Thus, we proved  $\mathbf{R}\vDash\Psi $.

Finally, let $(\mathbf{R}, \boldsymbol{\Gamma})$ be a maximal instance of CSP and $\mathbf{R} = (V,r_1, ..., r_s)$.
By the definition of the maximal instance, we have $r_i=\bigcap_{h\in {\rm Hom}(\mathbf{R}, \boldsymbol{\Gamma})} h^{-1}(\varrho_i)$. Horn formulas have the following simple property: if $(V,r^1_1, ..., r^1_s)\vDash\Psi $ and $(V,r^2_1, ..., r^2_s)\vDash\Psi $, then $(V,r^1_1\cap r^2_1, ..., r^1_s\cap r^2_s)\vDash\Psi $.
Since $(V, h^{-1}(\varrho_1), ..., h^{-1}(\varrho_s))\vDash\Psi $
for any $h\in {\rm Hom}(\mathbf{R}, \boldsymbol{\Gamma})$, we conclude $\mathbf{R}\vDash\Psi $.
\end{proof}

Theorem~\ref{closure} motivates the following approach to the problem ${\rm Dense}(\boldsymbol{\Gamma})$. Let $L=\{\Psi_1, ..., \Psi_c\}$ be a finite set of Horn formulas such that $\boldsymbol{\Gamma}\vDash\Psi_i$, $i\in [c]$. Given an instance $\mathbf{R} = (V,r_1, ..., r_s)$ of ${\rm Dense}(\boldsymbol{\Gamma})$, let us define an operator 
$$
q_i(r_1, ..., r_s) = r_i \cup\hspace{-20pt}\bigcup_{\Psi\in L: \Psi = \forall x_{1:n_i}(\Phi(x_1, ..., x_{n_i})\rightarrow \pi_i(x_1, ..., x_{n_i}))}\Phi^{\mathbf{R}},
$$
called the immediate consequence operator, i.e. it outputs a single application of the rules that contain $\pi_i$ as the head. This induces an operator on relational structures:
$$
Q(\mathbf{R}) = (V, q_1(r_1, ..., r_s), ..., q_s(r_1, ..., r_s))
$$
Since $q_i(r_1, ..., r_s)\supseteq r_i$, the Algorithm~\ref{naive} eventually stops at the fixed point of the operator $Q(\mathbf{R})$, i.e. at $Q^{K-1}(\mathbf{R})$ where: 
\begin{equation}\label{naive}
\mathbf{R}^0 = \mathbf{R}, \mathbf{R}^k =  Q (\mathbf{R}^{k-1}), k\in [K], \mathbf{R}^K = \mathbf{R}^{K-1}.
\end{equation}
In that algorithm we iteratively add new tuples to predicates $r_i, i\in [s]$ until all Horn formulas in $L$ are satisfied.

Let us denote the output $Q^{K-1}(\mathbf{R})$ of the Algorithm~\ref{naive} by $\mathbf{R}^L = (V,r^L_1, ..., r^L_s)$. In fact, the Algorithm~\ref{naive} calculates the fixed point of the operator $Q(\mathbf{R})$ in $O(|\mathbf{R}^L|)$ iterations, where $|\mathbf{R}^L| = \sum_{i=1}^s |r^L_i|$.
It is easy to see that $\mathbf{R}^L = (V,r^L_1, ..., r^L_s)$ is a smallest (w.r.t. inclusion) relational structure $\mathbf{T} = (V,t_1, ..., t_s)$ such that $t_i\supseteq r_i, i\in [s]$ and $\mathbf{T} \vDash\Psi_i$, $i\in [c]$. Therefore, $\mathbf{R}^L$ is a good candidate for a maximal instance $(\mathbf{R}'= (V,r'_1, ..., r'_s), \boldsymbol{\Gamma})$, $r'_i\supseteq r_i, i\in [s]$. 

\begin{definition} Let $\tau$ be a vocabulary and ${\rm F}\notin \tau$ be a stop symbol with an arity 0 assigned to it. Let $L$ be a finite set of Horn formulas over $\tau$ such that $\boldsymbol{\Gamma}\models \Psi, \Psi\in L$ and $L^{\rm stop}$ be a finite set of formulas of the form $\Phi\to {\rm F}$ where $\Phi$ is a quantifier-free primitive positive formula over $\tau$. It is said that ${\rm Dense}(\boldsymbol{\Gamma})$ can be solved by the Datalog program $L\cup L^{\rm stop}$, if for any instance $\mathbf{R}$ of ${\rm Dense}(\boldsymbol{\Gamma})$, we have: (a) if ${\rm Hom} (\mathbf{R}, \boldsymbol{\Gamma})\ne \emptyset$, then $(\mathbf{R}^L, \boldsymbol{\Gamma})$ is maximal and $\Phi^{\mathbf{R}^L}= \emptyset$ for any $(\Phi\to {\rm F})\in L^{\rm stop}$, and (b) if ${\rm Hom} (\mathbf{R}, \boldsymbol{\Gamma}) = \emptyset$, then there is $(\Phi\to {\rm F})\in L^{\rm stop}$ such that $\Phi^{\mathbf{R}^L}\ne \emptyset$.
\end{definition}

\begin{theorem} \label{DS-basis}
If ${\rm Dense}(\boldsymbol{\Gamma})$ can be solved by the Datalog program $L\cup L^{\rm stop}$, then $\boldsymbol{\Gamma}$ is a DS-basis.
\end{theorem}
\begin{proof}
Any $\Psi\in L$ can be represented as $$\Psi = \forall x_1, ..., x_{n} \big( \bigwedge_{t=1}^N \pi_{j_t} (x_{o_{t1}}, x_{o_{t2}}, ..., x_{o_{tn_{j_t}}})\rightarrow \pi_u (x_1, ..., x_{n_u})\big). $$ For any sequence $v_1, ..., v_n\in V$ let us introduce an implication
\begin{equation}\label{d-rules}
R_\Psi(v_1, ..., v_n) \rightarrow \langle (v_1, ..., v_{n_u}), \varrho_{u}\rangle
\end{equation}
where $R_\Psi(v_1, ..., v_n) = \big\{ 
\langle (v_{o_{t1}}, v_{o_{t2}}, ..., v_{o_{tn_{j_t}}}), \varrho_{j_t}\rangle | t\in [N]
\big\} \subseteq \mathcal{C}^{\boldsymbol{\Gamma}}_V$. Analogously, 
any $\Psi\in L^{\rm stop}$ can be represented as $\Psi = \big( \bigwedge_{t=1}^N \pi_{j_t} (x_{o_{t1}}, x_{o_{t2}}, ..., x_{o_{tn_{j_t}}})\rightarrow {\rm F}\big)$ and we define an implication
\begin{equation}\label{d-rules2}
R_\Psi(v_1, ..., v_n) \rightarrow {\mathcal C}^{\boldsymbol{\Gamma}}_V
\end{equation}
where $R_\Psi(v_1, ..., v_n) = \big\{ 
\langle (v_{o_{t1}}, v_{o_{t2}}, ..., v_{o_{tn_{j_t}}}), \varrho_{j_t}\rangle | t\in [N]
\big\} \subseteq \mathcal{C}^{\boldsymbol{\Gamma}}_V$.

 Let us denote
\begin{equation}\label{omega}
\begin{split}
\Omega^V_\Psi = \bigcup_{v_1, ..., v_n\in V}\{R_\Psi(v_1, ..., v_n) \rightarrow \langle (v_1, ..., v_{n_u}), \varrho_{u}\rangle\}
\end{split}
\end{equation}
if $\Psi\in L$ and 
\begin{equation*}
\begin{split}
\Omega^V_\Psi = \bigcup_{v_1, ..., v_n\in V}\{R_\Psi(v_1, ..., v_n) \rightarrow \mathcal{C}^{\boldsymbol{\Gamma}}_V\}
\end{split}
\end{equation*}
if $\Psi\in L^{\rm stop}$
and set
$$
\Sigma = \bigcup_{\Psi\in L\cup L^{\rm stop}} \Omega^V_\Psi
$$

Let us first prove the inclusion $\Sigma^\triangleright  \subseteq \Delta_1 \cup \Delta_2$ where $$\Delta_1 = \{\mathcal{C}_{\mathbf{R}}\rightarrow B | B\subseteq \mathcal{C}_{\mathbf{R}^L}, {\rm Hom}(\mathbf{R}, \boldsymbol{\Gamma}) \ne \emptyset\}$$ and $$ \Delta_2 = \{\mathcal{C}_{\mathbf{R}}\rightarrow B | B\subseteq \mathcal{C}^{\boldsymbol{\Gamma}}_V, {\rm Hom}(\mathbf{R}, \boldsymbol{\Gamma})= \emptyset\}. $$ For this, it is enough to show that $\Delta_1 \cup \Delta_2$ is a full implicational system and $\Sigma\subseteq \Delta_1 \cup \Delta_2$. The mapping ${\textsc O}: 2^{\mathcal{C}^{\boldsymbol{\Gamma}}_V}\rightarrow 2^{\mathcal{C}^{\boldsymbol{\Gamma}}_V}$, defined by ${\textsc O}(\mathcal{C}_{\mathbf{R}}) = \mathcal{C}_{\mathbf{R}^L}$ if ${\rm Hom}(\mathbf{R}, \boldsymbol{\Gamma})\ne \emptyset$ and ${\textsc O}(\mathcal{C}_{\mathbf{R}}) = \mathcal{C}^{\boldsymbol{\Gamma}}_V$ if ${\rm Hom}(\mathbf{R}, \boldsymbol{\Gamma}) = \emptyset$, is the closure operator by its construction. Therefore, Theorem~\ref{implicational} implies that the set $\Delta_1 \cup \Delta_2$ is a full implicational system.
The fact $\Sigma\subseteq \Delta_1 \cup \Delta_2$ is obvious, because for any rule of the form~\eqref{d-rules}, there exists an instance $\mathbf{R}$ such that $\mathcal{C}_{\mathbf{R}} = \{ \langle (v_{o_{t1}}, v_{o_{t2}}, ..., v_{o_{tn_{j_t}}}), \varrho_{j_t}\rangle | t\in [N] \}$. The naive evaluation algorithm~\ref{naive} will put the tuple $(v_1, ..., v_{n_u})$ into $r_u$ at the first iteration, because $(v_1, ..., v_{n_u})\in q_u(\mathbf{R})$. 
Thus, the head of that rule $\langle (v_1, ..., v_{n_u}), \varrho_{u}\rangle$ will be in $\mathcal{C}_{\mathbf{R}^L}$. Analogously, any rule of the form~\eqref{d-rules2} is also in $\Delta_1 \cup \Delta_2$.
Thus, we proved $\Sigma^\triangleright \subseteq \Delta_1 \cup \Delta_2$, and next we need to prove $\Delta_1 \cup \Delta_2\subseteq \Sigma^\triangleright$. 

Note that 
the operator $Q(\mathbf{R})$ operates on $\mathbf{R} = (V, r_1, ..., r_s)$ by computing tuples from $q_i(r_1, ..., r_s), i\in [s]$ in the following way: computing $(v_1, ...,  v_{n_i})\in q_i(r_1, ..., r_s)$ can be modeled as a result of applying one of the rules~\eqref{d-rules} to attributes from $\mathcal{C}_{\mathbf{R}}$ to obtain the attribute $\langle (v_1, ...,  v_{n_i}), \varrho_i\rangle$. Thus, $\mathcal{C}_{\mathbf{R}}\rightarrow \mathcal{C}_{Q(\mathbf{R})}\in \Sigma^\triangleright$. Therefore, $\mathcal{C}_{\mathbf{R}}\rightarrow \mathcal{C}_{Q^l(\mathbf{R})}\in \Sigma^\triangleright$ for any $l\in {\mathbb N}$, and we obtain $\mathcal{C}_{\mathbf{R}}\rightarrow \mathcal{C}_{\mathbf{R}^L}\in \Sigma^\triangleright$. Since $\Sigma^\triangleright$ is full, we conclude  $\{\mathcal{C}_{\mathbf{R}}\rightarrow B | B\subseteq \mathcal{C}_{\mathbf{R}^L}\}\subseteq \Sigma^\triangleright$. Moreover, if ${\rm Hom}(\mathbf{R}, \boldsymbol{\Gamma})= \emptyset$, we can prove that any rule $\mathcal{C}_{\mathbf{R}}\rightarrow B, B\subseteq \mathcal{C}^{\boldsymbol{\Gamma}}_V$ is in $\Sigma^\triangleright$. This implies $\Delta_1 \cup \Delta_2\subseteq \Sigma^\triangleright$.

In fact, we proved that the implicational system $\Sigma$ corresponds to the closure operator ${\textsc O}:2^{\mathcal{C}^{\boldsymbol{\Gamma}}_V}\rightarrow 2^{\mathcal{C}^{\boldsymbol{\Gamma}}_V}$ (defined before) with respect to the canonical correspondence of Theorem~\ref{implicational}.
The closure operator ${\textsc O}$ coincides with the densification operator ${\rm Dense}$. 

Thus, if ${\rm Dense}(\boldsymbol{\Gamma})$ can be solved by Datalog program $L$, then the implicational system $\Sigma$ satisfies $\Sigma^\triangleright = \Sigma_V^{\boldsymbol{\Gamma}}$ and $\boldsymbol{\Gamma}$ is a DS-basis.
\end{proof}

Obviously, if ${\rm Dense}(\boldsymbol{\Gamma})$ can be solved by some Datalog program $L \cup L^{\rm stop}$, then all the more $\neg {\rm CSP}(\boldsymbol{\Gamma})$ can be expressed by Datalog.
The following theorems give examples of constraint languages for which ${\rm Dense}(\boldsymbol{\Gamma})$ can be solved by Datalog.
\begin{theorem}\label{Horn-case} Let $\boldsymbol{\Gamma} = (D=\{0,1\}, \{(0)\}, \{(1)\}, \varrho_{x \wedge y\rightarrow z})$ where $\varrho_{x \wedge y\rightarrow z} = \{(a_1, a_2, a_{3})\in D^{3}| a_1 a_{2}\leq a_{3}\}$.
Then, there is a finite set of Horn formulas $L$ over $\tau = \{\pi_1, \pi_2, \pi_3\}\cup \{{\rm F}\}$ such that ${\rm Dense}(\boldsymbol{\Gamma})$ can be solved by the Datalog program $L$.
\end{theorem}

\begin{theorem}\label{2-SAT} Let $\boldsymbol{\Gamma} = (D=\{0,1\}, \varrho_1, \varrho_2, \varrho_3)$ where $ \varrho_1 = \big \{ ({x, y} ) | x \vee y \big \} $,
 $ \varrho_2 = \big \{ ({x, y} ) | \neg x \vee y \big \} $
 and $ \varrho_3 = \big \{{ ({x, y} ) | \neg x \vee \neg y} \big \} $. Then, there is a finite set of Horn formulas $L$ over $\tau = \{\pi_1, \pi_2, \pi_3\}\cup \{{\rm F}\}$ such that ${\rm Dense}(\boldsymbol{\Gamma})$ can be solved by the Datalog program $L$.
\end{theorem}

Proof of Theorem~\ref{Horn-case} is given in Section~\ref{proof-of-horn} and proof of Theorem~\ref{2-SAT} is given in Section~\ref{2-SAT-proof}.

\section{Classification of ${\rm Dense}(\boldsymbol{\Gamma})$ for the Boolean case}\label{Boolean-dense}

The problem ${\rm Dense}(\boldsymbol{\Gamma})$ is tightly connected with the so-called implication and equivalence problems, parameterized by $\boldsymbol{\Gamma}$.
\begin{definition}\label{impl-define} Let $\boldsymbol{\Gamma} = (D, \varrho_1, ..., \varrho_s)$. The {\bf implication problem}, denoted ${\rm Impl}(\boldsymbol{\Gamma})$,  is a decision task with:
\begin{itemize}
\item {\bf An instance:} two relational structures $\mathbf{R} = (V,r_1, ..., r_s)$ and $\mathbf{R}' = (V,r'_1, ..., r'_s)$.
\item {\bf An output:} yes, if ${\rm Hom}(\mathbf{R}, \boldsymbol{\Gamma}) \subseteq {\rm Hom}(\mathbf{R}', \boldsymbol{\Gamma})$, and no, if otherwise.
\end{itemize}
\end{definition}
Theorem 6.5 from~\cite{Schnoor2008} (which is based on the earlier result~\cite{Hemaspaandra}) gives a complete classification of the computational complexity of ${\rm Impl}(\boldsymbol{\Gamma})$ for Boolean languages.

\begin{theorem}[Schnoor, Schnoor, 2008]\label{Schnoor} If $\Gamma$ is Schaefer, then ${\rm Impl}(\boldsymbol{\Gamma})$ can be solved in polynomial time. Otherwise, it is coNP-complete under logspace reductions.
\end{theorem}

This theorem directly leads us to the classification of ${\rm Dense}(\boldsymbol{\Gamma})$.
\begin{theorem}\label{using-schnoor} If  $\Gamma$ is Schaefer, then ${\rm Dense}(\boldsymbol{\Gamma})$ is polynomially solvable. Otherwise, it is NP-hard.
\end{theorem}
\begin{proof} Let us show that ${\rm Dense}(\boldsymbol{\Gamma})$ can be solved in polynomial time using an oracle access to ${\rm Impl}(\boldsymbol{\Gamma})$. Indeed, let $\mathbf{R} = (V,r_1, ..., r_s)$ be an instance of ${\rm Dense}(\boldsymbol{\Gamma})$. Then, $\langle (v_1,\cdots, v_d), \varrho \rangle\in \mathcal{C}_V^{\boldsymbol{\Gamma}}$ is in ${\rm Dense}(\mathcal{C}_{\mathbf{R}})$ if and only if ${\rm Hom}(\mathbf{R}, \boldsymbol{\Gamma}) \subseteq {\rm Hom}(\mathbf{R}', \boldsymbol{\Gamma})$ where $\mathcal{C}_{\mathbf{R}'} = \{\langle (v_1,\cdots, v_d), \varrho \rangle\}$. Thus, by giving $(\mathbf{R},\mathbf{R}')$ to an oracle of ${\rm Impl}(\boldsymbol{\Gamma})$, we decide whether $\langle (v_1,\cdots, v_d), \varrho \rangle\in  {\rm Dense}(\mathcal{C}_{\mathbf{R}})$. By doing this for all  $\langle (v_1,\cdots, v_d), \varrho \rangle\in \mathcal{C}_V^{\boldsymbol{\Gamma}}$, we compute the whole set ${\rm Dense}(\mathcal{C}_{\mathbf{R}})$ in polynomial time.

Thus, ${\rm Dense}(\boldsymbol{\Gamma})$ is polynomial time Turing reducible to ${\rm Impl}(\boldsymbol{\Gamma})$, and therefore, using Theorem~\ref{Schnoor}, is polynomially solvable if $\Gamma$ is Schaefer.

Let us now show that $\neg {\rm Impl}(\boldsymbol{\Gamma})$ is polynomial time Turing reducible to ${\rm Dense}(\boldsymbol{\Gamma})$. Given an instance $(\mathbf{R}, \mathbf{R}')$ of $\neg {\rm Impl}(\boldsymbol{\Gamma})$, the inclusion ${\rm Hom}(\mathbf{R}, \boldsymbol{\Gamma}) \subseteq {\rm Hom}(\mathbf{R}', \boldsymbol{\Gamma})$ holds if and only if  $\mathcal{C}_{\mathbf{R}'}\subseteq {\rm Dense}(\mathcal{C}_{\mathbf{R}})$.
Thus, by computing ${\rm Dense}(\mathcal{C}_{\mathbf{R}})$ one can efficiently decide whether $\mathcal{C}_{\mathbf{R}'}\subseteq {\rm Dense}(\mathcal{C}_{\mathbf{R}})$, i.e. whether ${\rm Hom}(\mathbf{R}, \boldsymbol{\Gamma}) \subseteq {\rm Hom}(\mathbf{R}', \boldsymbol{\Gamma})$. If $\mathcal{C}_{\mathbf{R}'}\not\subseteq {\rm Dense}(\mathcal{C}_{\mathbf{R}})$, our reduction outputs "yes", and it outputs "no", if otherwise.

If $\Gamma$ is not Schaefer, then $\neg {\rm Impl}(\boldsymbol{\Gamma})$ is NP-complete, and therefore,  ${\rm Dense}(\boldsymbol{\Gamma})$ is NP-hard.
\end{proof}

\section{Non-Schaefer languages and mP/poly}
For our proof of Theorem~\ref{main-bounded}, we need to show  ${\rm Dense}(\boldsymbol{\Gamma})\notin {\rm mP/poly}$ for non-Schaefer languages. Note that under ${\rm NP}\not\subseteq {\rm P/poly}$ (which is widely believed to be true), any NP-hard problem is outside of P/poly. Therefore, if $\Gamma$ is not Schaefer, then ${\rm Dense}(\boldsymbol{\Gamma})\notin {\rm P/poly}$ (and all the more, ${\rm Dense}(\boldsymbol{\Gamma})\notin {\rm mP/poly}$). In the current section we prove ${\rm Dense}(\boldsymbol{\Gamma})\notin {\rm mP/poly}$ {\em unconditionally} and this fact will be used in Section~\ref{bounded-width-main}.
\begin{theorem}\label{non-schaefer} Let $\Gamma$ be a non-Schaefer language. Then, ${\rm Dense}(\boldsymbol{\Gamma})\notin {\rm mP/poly}$.
\end{theorem}
\begin{lemma}\label{equiv} For any language $\Gamma$ that is not constant preserving, if ${\rm Dense}(\boldsymbol{\Gamma})\in {\rm mP/poly}$, then $\neg {\rm CSP}(\boldsymbol{\Gamma})\in {\rm mP/poly}$. 
\end{lemma}
\begin{proof}
Let $\mathbf{R}$ be an instance of ${\rm CSP}(\boldsymbol{\Gamma})$ and $\{x_v\}_{v\in \mathcal{C}_V^{\boldsymbol{\Gamma}}}$ be input boolean variables of a monotone polynomial-size circuit that computes ${\rm Dense}(\mathcal{C}_{\mathbf R})$ such that $x_v = 1$ if and only if $v\in \mathcal{C}_{\mathbf R}$. Let $\{y_{v}\}_{v\in \mathcal{C}_V^{\boldsymbol{\Gamma}}}$ be output variables of that circuit, and $y_v=1$ indicates $v\in {\rm Dense}(\mathcal{C}_{\mathbf R})$. Then, $\bigwedge\limits_{v\in \mathcal{C}_V^{\boldsymbol{\Gamma}}}y_v=1$ if and only if ${\rm Hom}(\mathbf{R}, \boldsymbol{\Gamma})= \emptyset$. Thus, emptiness of ${\rm Hom}(\mathbf{R}, \boldsymbol{\Gamma})$ can be decided by a polynomial-size monotone circuit. Therefore, $\neg {\rm CSP}(\boldsymbol{\Gamma})\in {\rm mP/poly}$.

\end{proof}


In the case $D = \left\{ {0,1} \right\}$, there is a countable number
of clones: in the list below we use the notation from the table on page 76 of~\cite{marchenkov} (the same results can be found in the table on page 1402 of~\cite{Mar98}), together with the notation from the Table 1 of~\cite{elmar}. 
For every row,  listed relations form a basis of the relational co-clone corresponding to the functional clone (notations of clones are given according to~\cite{marchenkov} and~\cite{elmar}). At the same time, the functional clone equals the set of polymorphisms of the relations. Below we list all Post co-clones $\Gamma$ except for those that: a) satisfy $\{(0)\}, \{(1)\}\in \Gamma$ and b) $\Gamma$ is Schaefer  (and we are not interested in such languages in the current section).
\begin{equation}\label{post-table}
\begin{array}{ |c|c|c|c| }
\hline
~\cite{marchenkov} &  $\cite{elmar}$ & {\rm Basis \,\,of \,\,clone} & {\rm Basis \,\,of \,\,co-clone} \\
\hline
   {U } & N & \{\neg x,1\} & \varrho_{\rm b}=\big\{\{(x_1,x_2,x_3)\mid x_1  = x_2  \vee x_2  = x_3\} \big\}  \\
   {SU } & N_2 & \{\neg x\} & \varrho_{\rm NAE}=\big\{\{(x_1,x_2,x_3)\mid x_1 \ne x_2  \vee x_1 \ne x_3\} \big\}  \\
   {MU } & I& \{0,1\} & \big\{\{(x_1,x_2)\mid x_1  \le x_2\}, \varrho_{\rm b}\big\}  \\
   {U_{0} } & I_0& \{0\} & \big\{\{(0)\},\varrho_{\rm b} \big\}  \\
   {U_{1} } & I_1& \{1\} & \big\{\{(1)\},\varrho_{\rm b} \big\}  \\
   \hline
\end{array}
\end{equation}
Next, we will concentrate on languages listed in Table~\ref{post-table}.

Our first goal is to study the complexity of ${\rm Dense}(\boldsymbol{\Gamma})$ where $\boldsymbol{\Gamma} = (\{0,1\}, \varrho_{\rm b})$ where $\varrho_{\rm b} = \{(x_2, x_1, x_3)| x_1  = x_2  \vee x_1  = x_3\}$. 

\begin{lemma}\label{between}
${\rm Dense}(\boldsymbol{\Gamma} = (\{0,1\}, \varrho_{\rm b}))\notin {\rm mP/poly}$.
\end{lemma}
\begin{proof}
Let us introduce the restriction of ${\rm CSP}(\boldsymbol{\Gamma})$, $\boldsymbol{\Gamma} = ( \{0,1\}, \varrho_{\rm b}, \{(0)\}, \{(1)\} )$, in which we assume that in its instance ${\mathbf R} =(V, r, \{Z\}, \{O\})$ the domain $V$ contains two designated variables, $Z$ and $O$, with unary constraints, $Z=0$ and $O=1$. This task is denoted by ${\rm CSP}_{\rm b}$.

It is easy to see that 
\begin{equation*}
\begin{split}
\varrho_{\rm NAE}(x,y,z) = \exists t, O, Z\,\, \varrho_{\rm b}(x,t,z) \wedge 
\varrho_{\rm b}(t,Z,y) \wedge \varrho_{\rm b}(t,O,y) \wedge [O=1] \wedge [Z=0]
\end{split}
\end{equation*}
where $\varrho_{\rm NAE} = \{(x_1,x_2,x_3)| x_1 \ne x_2  \vee x_1 \ne x_3 \}$. Thus, by  ${\rm CSP}_{\rm b}$ we can model any instance of ${\rm CSP}(\{\varrho_{\rm NAE}\})$. 
The standard reduction  of ${\rm CSP}(\{\varrho_{\rm NAE}\})$ to ${\rm CSP}_{\rm b}$ can be implemented as a monotone circuit. Since $\{\varrho_{\rm NAE},\{(0)\},\{(1)\}\}$ is not of bounded width and $\neg {\rm CSP}(\{\varrho_{\rm NAE}\})$ is equivalent to $\neg {\rm CSP}(\{\varrho_{\rm NAE},\{(0)\},\{(1)\}\})$ modulo polynomial-size reductions by monotone circuits (see analogous argument in the proof of Corollary~\ref{corol-mon}), we conclude $\neg {\rm CSP}(\{\varrho_{\rm NAE}\})\notin {\rm mP/poly}$ (using Proposition 5.1. from~\cite{ability_to_count}). Therefore, $\neg {\rm CSP}_{\rm b}\notin {\rm mP/poly}$.

Let us now prove that ${\rm Dense}(\boldsymbol{\Gamma})$, where $\boldsymbol{\Gamma}=(\{0,1\}, \varrho_{\rm b})$, is outside of mP/poly. 
Let ${\mathbf R} =(V, r)$ be an instance of ${\rm Dense}(\boldsymbol{\Gamma} = (\{0,1\}, \varrho_{\rm b}))$ and let ${\mathbf R}'  =(V, r)$ be such that $r'\supseteq r$ and $({\mathbf R}', \boldsymbol{\Gamma})$ is a maximal instance. By construction, for any $i,j\in V$, $(i,j,i)\in r'$ if and only if there is no such $h\in {\rm Hom}({\mathbf R}, \boldsymbol{\Gamma} )$ that satisfies $h(i)=0$ and $h(j)=1$. But the last question, i.e. checking the emptiness of $\{h\in {\rm Hom}({\mathbf R}, \boldsymbol{\Gamma} ) | h(i)=0, h(j)=1\}$ is equivalent to $\neg {\rm CSP}_{\rm b}$ after setting $Z=i, O=j$. The latter argument can be turned into a reduction of $\neg {\rm CSP}_{\rm b}$ to ${\rm Dense}(\boldsymbol{\Gamma} = (\{0,1\}, \varrho_{\rm b}))$. Again, this reduction can be implemented as a monotone circuit.

Therefore, ${\rm Dense}(\boldsymbol{\Gamma} = (\{0,1\}, \varrho_{\rm b}))\notin {\rm mP/poly}$.
\end{proof}

\begin{lemma}\label{null} If $\langle\Gamma\rangle$ equals one of $ {\rm inv} (U_0),  {\rm inv} (U_1),  {\rm inv} (SU), {\rm inv} (M U)$ and $ {\rm inv} (U)$, then $\varrho_{\rm b}$ is strongly reducible to $\Gamma$.
\end{lemma}
\begin{proof}
Let $\Gamma = \{\rho_1, \cdots, \rho_s\}$. Since $\varrho_{\rm b}\in {\rm inv}(U)\subseteq {\rm inv} (U_0),  {\rm inv} (U_1),  {\rm inv} (SU), {\rm inv} (M U)$, then $\varrho_{\rm b} = \Psi^{\boldsymbol{\Gamma}}$  for a primitive positive formula $\Psi$ over $\tau = \{\pi_1, \cdots, \pi_s\}$. Let $$\Psi = \exists x_{4}... x_{l} \bigwedge_{t\in [N]}\pi_{j_{t}} (x_{o_{t1}}, x_{o_{t2}}, ...).$$ Let us denote $\Phi = \bigwedge_{t\in [N]}\pi_{j_{t}} (x_{o_{t1}}, x_{o_{t2}}, ...)$ and consider a relation $\gamma=\{{\mathbf x}\in \{0,1\}^l\mid {\mathbf x}\in \Phi^{\boldsymbol{\Gamma}}\,\,{\rm or}\,\,{\mathbf x}_{1:3}\notin \varrho_{\rm b}\}$. Let us prove that if $u\in {\rm pol}(\Phi^{\boldsymbol{\Gamma}})$ and $u$ is unary, then $u\in  {\rm Pol}(\Gamma)$. The latter can be checked by considering all 4 cases: $u(x)=x$, or $\neg x$, or 0, or 1. A unary $u(x)=x$ is a polymorphism of any relation. If $u(x)=c$, then $u\in {\rm pol}(\Phi^{\boldsymbol{\Gamma}})$ means that $\Phi^{\boldsymbol{\Gamma}}$ is a $c$-preserving relation. Then $\gamma$ is also $c$-preserving. Finally, if $u(x)=\neg x$, then $u\in {\rm pol}(\Phi^{\boldsymbol{\Gamma}})$ means that $\Phi^{\boldsymbol{\Gamma}}$ is a self-dual relation. Therefore, $\gamma = \Phi^{\boldsymbol{\Gamma}} \cup \{(0,1,0), (1,0,1)\}\times D^{l-3}$ is also self-dual, i.e. $u\in {\rm Pol}(\Gamma)$.

From the last fact we conclude that $\{u: D\to D \mid u\in {\rm Pol}(\Gamma)\}\subseteq \{u: D\to D \mid u\in {\rm pol}(\{\gamma\})\}$. Since  $\{u: D\to D \mid u\in {\rm Pol}(\Gamma)\}$ forms a basis of ${\rm Pol}(\Gamma)$ (in all listed cases), then $\gamma\in {\rm inv}({\rm Pol}(\Gamma))$, i.e. $\gamma\in \langle\Gamma\rangle$. 

Finally, by construction we have $\gamma = \Phi^{\boldsymbol{\Gamma}}\cup \delta$ where ${\rm pr}_{1,2,3}\delta = D^3\setminus \varrho_{\rm b}$ and ${\rm pr}_{1,2,3} 
\Phi^{\boldsymbol{\Gamma}}=\varrho_{\rm b}$. This is exactly the needed condition for $\varrho_{\rm b}$ to be strongly reducible to $\Gamma$.
\end{proof}
\begin{proof}[Theorem~\ref{non-schaefer}]
We have $D = \{0,1\}$. Our goal is to prove that
if $\Gamma$ is non-Schaefer then ${\rm Dense}(\boldsymbol{\Gamma})$ is outside of mP/poly.


Let us first consider the subcase where ${\rm CSP}(\boldsymbol{\Gamma})$ is NP-hard. Then, by construction, $\Gamma$ is not constant preserving and ${\rm core}(\Gamma)=\Gamma$. Therefore, $\neg {\rm CSP}(\Gamma)$ and  $\neg {\rm CSP}(\Gamma\cup \{(0)\}\cup \{(1)\})$ can be mutually reduced by polynomial-size monotone circuits (as in the proof of corollary~\ref{corol-mon}). Since $\neg {\rm CSP}(\Gamma)\equiv \neg {\rm CSP}(\Gamma\cup \{(0)\}\cup \{(1)\})\notin {\rm mP/poly}$ (by proposition 5.1. from~\cite{ability_to_count}), then, by Lemma~\ref{equiv}, ${\rm Dense}(\boldsymbol{\Gamma})\notin {\rm mP/poly}$.

Next, let us consider the subcase where ${\rm CSP}(\boldsymbol{\Gamma})$ is tractable. Since we already assumed that $\Gamma$ is not any of 4 Schaeffer classes, this can happen only if $\Gamma$ is constant preserving. Therefore, $\{\{0\}, \{1\}\}\not\subseteq \langle\Gamma\rangle$. All possible variants for ${\rm Pol}(\Gamma)$ are listed in table~\ref{post-table}. 
Since $\langle \{\varrho_{\rm b}\}\rangle = {\rm inv}(U)\subseteq {\rm inv} (U_0),  {\rm inv} (U_1),  {\rm inv} (SU), {\rm inv} (M U)$,  Lemma~\ref{between} in combination with Lemma~\ref{null} and part (c) of Theorem~\ref{reduction} gives us that  ${\rm Dense}(\boldsymbol{\Gamma})\notin {\rm mP/poly}$. 

\end{proof}

\section{Proof of Theorem~\ref{main-bounded}}\label{bounded-width-main}
Let us prove first that for the Boolean domain $D=\{0,1\}$, if  $\Gamma$ satisfies one of the following 3 conditions
\begin{itemize}
\item[(a)] $\Gamma$ is a subset of $\langle \{\varrho_1, \varrho_2, \varrho_3\}\rangle $ where $ \varrho_1 = \left \{{\left ({x, y} \right) | x \vee y} \right \} $, $ \varrho_2 = \left \{{\left ({x, y} \right) | \neg x \vee y} \right \} $ and $ \varrho_3 = \left \{{\left ({x, y} \right) | \neg x \vee \neg y} \right \} $ (2-SAT);
\item[(b)] $\Gamma$ is a subset of $\langle \{\{(0)\}, \{(1)\}, \varrho_{x \wedge y\rightarrow z}\}\rangle $ (Horn case); 
\item[(c)] $\Gamma$ is a subset of $\langle \{\{(0)\}, \{(1)\}, \varrho_{\neg x \wedge \neg y\rightarrow \neg z}\}\rangle $ (dual-Horn case). 
\end{itemize}
then it has a  weak polynomial densification operator. 

Note that from Theorems~\ref{SAT-case-reduce} and~\ref{Horn-case-reduce} it follows that  in all three cases $\Gamma$  is a subset of the
relational clone of an A-language. Part (b) of Theorem~\ref{reduction} claims that $\Gamma$ has a weak polynomial densification operator if languages $\{\varrho_1, \varrho_2, \varrho_3\}, \{\{(0)\}, \{(1)\}, \varrho_{x \wedge y\rightarrow z}\}$ have one. Theorems~\ref{DS-basis}, \ref{Horn-case} and~\ref{2-SAT} give us that $(D,\varrho_1, \varrho_2, \varrho_3), (D,\{(0)\}, \{(1)\}, \varrho_{x \wedge y\rightarrow z})$ are DS-templates.
Therefore,  $\Gamma$ has a weak polynomial densification operator. 

It remains to prove that, in the Boolean case, the weak polynomial densification property implies one of these 3 conditions.

For the general domain $D$, if a constraint language 
$\Gamma$ has a weak polynomial densification operator, then its core is of bounded width (Theorem~\ref{no-linear}). Thus, in the Boolean case, if $\Gamma$ is not constant-preserving and has a weak polynomial densification operator, then it is of bounded width (i.e. $\Gamma$ is in one of the latter three classes). If $\Gamma$ preserves some constant $c$, then w.l.o.g. we can assume that $c=0$. From Theorem~\ref{Boolean-dense-thm}, whose proof is given in Section~\ref{Boolean-dense}, it is clear that either a) ${\rm Dense}(\boldsymbol{\Gamma})$ is NP-hard, or b)  $\Gamma$ is Schaefer, i.e. $\{\{0\}, \{1\}\}\cup \Gamma$ is tractable. In the first case, existence of a polynomial-size implicational system for the densification operator implies that there exists a monotone circuit of size ${\rm poly}(|V|)$ that computes the densification operator ${\rm Dense}$ (a construction of such a circuit is identical to the one given in Theorem~\ref{no-linear}). In other words, ${\rm Dense}(\boldsymbol{\Gamma})\in {\rm mP/poly}$. This contradicts to the claim of Theorem~\ref{non-schaefer} that ${\rm Dense}(\boldsymbol{\Gamma})\notin {\rm mP/poly}$ for non-Schaefer languages.

Thus, we have option b), and this can happen only if either b.1) $\Gamma$ preserves  $\vee$, or $\wedge$, or ${\rm mjy}(x,y,z) = (x\wedge y)\vee (x\wedge z)\vee (y\wedge z)$, or b.2) $\Gamma$ preserves $x\oplus y \oplus z$, but does not preserve $\vee, \wedge$ and ${\rm mjy}$. In the first case, $\Gamma$ satisfies the needed conditions. In the second case, $\Gamma$ is a 0-preserving language, i.e. $0, x\oplus y \oplus z\in {\rm Pol}(\Gamma)$, but $\vee, \wedge, {\rm mjy}\notin {\rm Pol}(\Gamma)$. According to table 2.1 on page 76 of Marchenkov's textbook~\cite{marchenkov}, there are only two functional clones with these properties, i.e. either b.2.1) ${\rm Pol}(\Gamma) = L$ where $L = \{a_0\oplus a_1x_1\oplus \cdots \oplus a_k x_k\}$ is a set of all linear functions, or b.2.2) ${\rm Pol}(\Gamma) = L_0$ where $L_0 = \{a_1x_1\oplus \cdots \oplus a_k x_k\}$. In both cases $\rho_L = \{(x,y,z,t)\mid x\oplus y\oplus z \oplus t=0\}\in \langle \Gamma\rangle$.

\begin{lemma} If ${\rm Pol}(\Gamma) = L_0$ or ${\rm Pol}(\Gamma) = L$, then $\rho_L$ is strongly reducible to $\Gamma$.
\end{lemma}
\begin{proof}
Note that $x\oplus y\in L_0\subseteq L$. Therefore, for any $\varrho\in \langle \Gamma \rangle$ we have $ \forall {\mathbf x}, {\mathbf y}\in \varrho\to {\mathbf x}\oplus {\mathbf y}\in \varrho$ where $\oplus$ is applied component-wise, i.e. $\varrho$ is a linear subspace. 
Since $\rho_L\in \langle \Gamma \rangle$, then there is a quantifier-free primitive positive formula $\Phi(x_1, \cdots, x_l)$ such that $\rho_L = {\rm pr}_{1,2,3,4}\Phi^{\boldsymbol{\Gamma}}$. Let us set $\Psi(x_1,\cdots, x_l) = \exists x_4 \Phi(x_1, \cdots, x_l)$, i.e. $\Psi$ depends on $x_4$ fictitiously. 
Let us define $\delta = \Psi^{\boldsymbol{\Gamma}}\setminus \Phi^{\boldsymbol{\Gamma}}$. Thus, we have $\Phi^{\boldsymbol{\Gamma}}\cup \delta\in  \langle \Gamma \rangle$, $\rho_L = {\rm pr}_{1,2,3,4}\Phi^{\boldsymbol{\Gamma}}$ and ${\rm pr}_{1,2,3,4}\delta ={\rm pr}_{1,2,3,4} \Psi^{\boldsymbol{\Gamma}}\setminus \Phi^{\boldsymbol{\Gamma}} = {\rm pr}_{1,2,3,4} \{{\mathbf x}\oplus a(0,0,0,1,0,\cdots,0)\mid a\in D, {\mathbf x}\in \Phi^{\boldsymbol{\Gamma}}\}\setminus \Phi^{\boldsymbol{\Gamma}} =  {\rm pr}_{1,2,3,4} \{{\mathbf x}\oplus (0,0,0,1,0,\cdots,0)\mid {\mathbf x}\in \Phi^{\boldsymbol{\Gamma}}\}= \{(x,y,z,t)\mid x\oplus y\oplus z \oplus t=1\} = D^4\setminus \rho_L$. The latter is the condition for strong reducibility of $\rho_L$ to $\Gamma$.
\end{proof}

Using part (b) of Theorem~\ref{reduction}, the weak polynomial densification property of $\Gamma$ and the latter lemma, we obtain that $\{\rho_L\}$ has a weak polynomial densification operator. The following Lemma contradicts to our conclusion. Therefore, in the Boolean case, the weak polynomial densification property implies one of 3 conditions given above.

\begin{lemma} $\{\rho_L\}$ does not have a weak polynomial densification operator.
\end{lemma}
\begin{proof} Let us prove the statement by reductio ad absurdum. Suppose that $\Gamma=\{\rho_L\}$ has  a weak polynomial densification operator. Therefore, ${\rm Dense}(\boldsymbol{\Gamma})\in {\rm mP/poly}$.

Since the core of $\Gamma'=\{\varrho, \{(0)\}, \{(1)\}\}$ where $\varrho = \{(x,y,z) \mid x\oplus y\oplus z=0\}$ is not of bounded width, by proposition 5.1 from~\cite{ability_to_count}, $\neg {\rm CSP}(\boldsymbol{\Gamma}')$ cannot be computed by a polynomial-size monotone circuit. Let us describe a monotone reduction of $\neg {\rm CSP}(\boldsymbol{\Gamma}')$ to ${\rm Dense}(\boldsymbol{\Gamma})$ which will imply $\neg {\rm CSP}(\boldsymbol{\Gamma}')\in {\rm mP/poly}$. This will be a contradiction.

According to~\cite{marchenkov}, $\Gamma=\{\rho_L\}$ is a basis of ${\rm Inv}(L_0)$. Therefore, 
$\langle \{\rho_L\}\rangle$ equals the set of all linear subspaces in $\{0,1\}^n, n\in {\mathbb N}$. In other words, for any $([n], r)$, ${\rm Hom}(([n], r),\boldsymbol{\Gamma})$ is a linear subspace of $\{0,1\}^n$, and $\{{\rm pr}_{[k]}{\rm Hom}(([n], r),\boldsymbol{\Gamma}) \mid n\in {\mathbb N}, k\leq n, r\subseteq [n]^4\}$ spans all possible linear subspaces. 

Let $\mathbf{R}' = ([n], r', Z, O)$ be an instance of $\neg {\rm CSP}(\boldsymbol{\Gamma}'=(D,\varrho, \{(0)\}, \{(1)\}))$.
Since $\varrho, \{(0)\}\in \langle \{\rho_L\}\rangle$,  a set of constraints $\{\langle (v_1,v_2,v_3), \varrho\rangle \mid (v_1,v_2,v_3)\in r'\}\cup \{\langle v, \{(0)\}\rangle \mid v\in Z\}\cup \{\langle (n+1,n+2,n+3),\varrho\rangle\}$  can be modeled as a set of constraints over $\{\rho_L\}$ with an extended set of variables $[m], m\geq n+3$, or alternatively, as 
an instance $\mathbf{R}'' = ([m], r)$ of ${\rm CSP}( \{\rho_L\})$.  Let $\sim$ 
be an equivalence relation on $[m+1]$ with equivalence classes $\big\{\{i\}\mid i \in [m]\setminus O\big\}\cup \big\{\{m+1\}\cup O\big\}$ and let $\overline{x}$ denote an equivalence class that contains $x\in [m+1]$.
A relational structure $\mathbf{R} = ([m+1]/\sim, \overline{r})$ where $\overline{r} = \{(\overline{x}, \overline{y}, \overline{z}, \overline{t}) \mid (x,y,z,t)\in r\}$, considered as an instance of ${\rm Dense}( \{\rho_L\})$, satisfies: $(\overline{n+1},\overline{n+2},\overline{n+3},\overline{m+1})\in {\rm Dense}(\mathcal{C}_{\mathbf{R}})$ if and only if ${\rm Hom}(\mathbf{R}', \boldsymbol{\Gamma}') = \emptyset$. Indeed, the constraint $\langle (n+1, n+2,n+3), \varrho\rangle$ that is satisfied for assignments in ${\rm Hom}(\mathbf{R}'', \boldsymbol{\Gamma})$, together with $(\overline{n+1},\overline{n+2},\overline{n+3},\overline{m+1})\in {\rm Dense}(\mathcal{C}_{\mathbf{R}})$, implies that $h(\overline{m+1})=0$ for any $h\in {\rm Hom}(\mathbf{R}, \boldsymbol{\Gamma})$. Or, equivalently, $\{h\in {\rm Hom}(\mathbf{R}, \boldsymbol{\Gamma}) \mid h(\overline{m+1})=1\}=\emptyset$. The latter is equivalent to $\{h\in {\rm Hom}(\mathbf{R}'', \boldsymbol{\Gamma}) \mid h(x)=1, x\in O\}=\emptyset$, or ${\rm Hom}(\mathbf{R}', \boldsymbol{\Gamma}')=\emptyset$.

By construction, the indicator Boolean vector of the subset $\mathcal{C}_{\mathbf{R}''}\in 2^{\mathcal{C}_{[m]}^{\boldsymbol{\Gamma}}}$, i.e. the Boolean vector $\mathbf{x}\in \{0,1\}^{\mathcal{C}_{[m]}^{\boldsymbol{\Gamma}}}$, $\mathbf{x}(\langle (v_1, v_2,v_3,v_4), \rho_L\rangle )=1\Leftrightarrow \langle (v_1, v_2,v_3,v_4), \rho_L\rangle \in \mathcal{C}_{\mathbf{R}''}$ can be computed from the indicator Boolean vector of $\mathcal{C}_{\mathbf{R}'}\in 2^{\mathcal{C}_{[n]}^{\boldsymbol{\Gamma}'}}$ by a polynomial-size monotone circuit.
Further, the indicator Boolean vector of the subset $\mathcal{C}_{\mathbf{R}}\in 2^{\mathcal{C}_{[m+1]/\sim}^{ \boldsymbol{\Gamma}}}$ can be computed by  a polynomial-size monotone circuit from the indicator Boolean vector of $\mathcal{C}_{\mathbf{R}''}\in 2^{\mathcal{C}_{[m]}^{\boldsymbol{\Gamma}}}$ and the indicator Boolean vector of $O\in 2^{[n]}$. Finally, we feed the indicator vector of $\mathcal{C}_{\mathbf{R}}$ to ${\rm Dense}(\boldsymbol{\Gamma})$ and compute whether $(\overline{n+1},\overline{n+2},\overline{n+3},\overline{m+1})\in {\rm Dense}(\mathcal{C}_{\mathbf{R}})$. Thus, the emptiness of ${\rm Hom}(\mathbf{R}', \boldsymbol{\Gamma}')$ can be decided by a polynomial-size monotone circuit which contradicts $\neg {\rm CSP}(\Gamma')\not\in {\rm mP/poly}$.

\end{proof}

\section{Proofs of Theorems~\ref{SAT-case-reduce} and~\ref{Horn-case-reduce}}\label{Horn-case-reduce-proof}
\begin{proof}[Theorem~\ref{SAT-case-reduce}] Let $\boldsymbol{\Gamma} = (D=\{0,1\}, \varrho_1, \varrho_2, \varrho_3)$ where $ \varrho_1 = \big \{ ({x, y} ) | x \vee y \big \} $,
 $ \varrho_2 = \big \{ ({x, y} ) | \neg x \vee y \big \} $
 and $ \varrho_3 = \big \{{ ({x, y} ) | \neg x \vee \neg y} \big \} $. 

First, let us note that  any binary relation $\rho\subseteq D^2$ is strongly reducible to $\Gamma$, due to $\rho = \bigcap_{\gamma\in S:  \rho\subseteq \gamma}\gamma$ where $S = \{\varrho_1,\varrho_2,\varrho_3, \varrho_2^T\}$, $\varrho_2^T = \{(y,x)\mid (x,y)\in \varrho_2\}$ (in the definition of strong reducibility one can set $\Xi(x,y) = \bigwedge_{i:\rho\subseteq \varrho_i} \pi_i(x,y) \bigwedge_{\rho\subseteq \varrho^T_2} \pi_2(y,x)$ and $\delta = D^2\setminus \rho$).

It is well-known that $\langle \Gamma\rangle = {\rm pol}({\rm mjy})$ where ${\rm mjy}(x,y,z) = (x\wedge y)\vee (x\wedge z)\vee (y\wedge z)$ is a majority operation. 
Every $n$-ary relation $\rho\in\langle \Gamma\rangle$ is defined by its binary projections $\rho_{ij} = \{(x_i,x_j)\mid (x_1, \cdots, x_n)\in \rho\}$, i.e.
$$
\rho = \bigcap_{ i,j \in [n]}r_{ij}
$$
where $r_{ij} = \{(x_1, \cdots, x_n) \mid (x_i,x_j)\in  \rho_{ij}\}$. Since $\rho_{ij}$ is strongly reducible to $\Gamma$, $r_{ij}$ also has this property. Thus, $\rho$ is A-reducible to $\Gamma$, and therefore, $\Gamma$ is an A-language.
\end{proof}

{\bf The Horn case.} Let $\boldsymbol{\Gamma} = (D=\{0,1\}, \{(0)\}, \{(1)\}, \varrho_{x \wedge y\rightarrow z})$.
In other words, $\langle \Gamma\rangle$ is a set of relations that is closed under component-wise conjunction, i.e. ${\mathbf x}, {\mathbf y}\in \rho\in \langle \Gamma\rangle$ implies ${\mathbf x}\wedge {\mathbf y}\in \rho$. 
\begin{lemma} Let $D=\{0,1\}$ and $\rho$ be a set of satisfying assignments of a Horn clause, i.e. $$\rho = \{(x_{1}, \cdots, x_{n})\mid (x_{1}\wedge \cdots\wedge x_{n}\to 0)\}$$ or $$\rho = \{(x_{1}, \cdots, x_{n+1})\mid  (x_{1}\wedge \cdots\wedge x_{n}\to x_{n+1})\}.$$ Then, $\rho$ is strongly reducible to $\Gamma$.
\end{lemma}
\begin{proof}
Let us consider first the case of
$
\Phi = (x_{1}\wedge \cdots\wedge x_{n}\to 0)$. This formula can be given as $\Phi \equiv \exists x_{n+1}, \cdots,x_{2n-1} \Xi(x_{1}, \cdots,x_{2n-1})$ where
$$
\Xi(x_{1}, \cdots,x_{2n-1}) =(x_1\wedge x_2\to x_{n+1}) \wedge (x_{2n-1}=0) \bigwedge_{i=3}^n (x_i\wedge x_{n+i-2}\to x_{n+i-1}).
$$
If we define a $2n-1$-ary $\delta$ as $\{(1,\cdots, 1)\}$, then it can be checked that $\Xi^{\boldsymbol{\Gamma}}\cup \delta$ is a $\wedge$-closed set. Indeed, for any ${\mathbf x}\in \Xi^{\boldsymbol{\Gamma}}$ and ${\mathbf y}\in \delta$, we have ${\mathbf x}\wedge {\mathbf y}={\mathbf x}\in \Xi^{\boldsymbol{\Gamma}}\cup \delta$. Since both  $\Xi^{\boldsymbol{\Gamma}}$ and $\delta$ are $\wedge$-closed, then we conclude the statement. Therefore, $\Xi^{\boldsymbol{\Gamma}}\cup \delta\in \langle \Gamma\rangle$. It remains to check that ${\rm pr}_{1:n}\Xi^{\boldsymbol{\Gamma}} = \rho$ and ${\rm pr}_{1:n}\delta = \{0,1\}^n\setminus \rho$. Thus, $\Xi^{\boldsymbol{\Gamma}}\cup \delta\in \langle \Gamma\rangle$ and $\rho = \{(x_{1}, \cdots, x_{n})\mid (x_{1}\wedge \cdots\wedge x_{n}\to 0)\}$ is strongly reducible to $\Gamma$.

Let us now consider the case of 
$
\Phi = (x_{1}\wedge \cdots\wedge x_{n}\to x_{n+1})
$. Let us denote by $(x\wedge y = z)$ the formula $(x\wedge y \to z)\wedge (z\wedge O \to x)\wedge (z\wedge O \to y)\wedge (O=1)$ where $O$ is an additional fixed variable. Note that $(x\wedge y = z)$ is a quantifier-free primitive positive formula over $\tau$.
Thus, we have $\Phi\equiv \exists x_{n+2}, \cdots,x_{2n-1},O\,\, \Xi(x_{1}, \cdots,x_{2n-1},O)$ where
\begin{equation*}
\begin{split}
\Xi(x_{1}, \cdots,x_{2n-1},O) =(x_1\wedge x_2 = x_{n+2}) \wedge 
(x_{n}\wedge x_{2n-1}\to x_{n+1}) \wedge
\bigwedge_{i=3}^{n-1} (x_i\wedge x_{n+i-1} = x_{n+i}).
\end{split}
\end{equation*}
Here we define a $2n$-ary $\delta$ as $\{1\}^{n}\times \{0\}\times \{1\}^{n-1}$. Let us prove that $\Xi^{\boldsymbol{\Gamma}}\cup \delta$ is a $\wedge$-closed set. Again, let us consider  ${\mathbf x}\in \Xi^{\boldsymbol{\Gamma}}$ and ${\mathbf y}\in \delta$. If $x_{n+1}=0$, then ${\mathbf x}\wedge {\mathbf y}={\mathbf x}\in \Xi^{\boldsymbol{\Gamma}}\cup \delta$. Otherwise, if $x_{n+1}=1$, we have either a) ${\mathbf x}=1^{2n-1}$ and in that case  $1^{2n-1}\wedge {\mathbf y}={\mathbf y}\in \Xi^{\boldsymbol{\Gamma}}\cup \delta$, or b) at least one of $x_1, \cdots, x_n$ is 0. In the case of b) let $i\in [n]$ be the smallest such that $x_i=0$, i.e. $x_j=1, j\in [i-1]$.  Therefore, $x_{n+j}=1, j\in [2,i-1]$ and $x_{n+j}=0, j\in [i,n-1]$. It remains to check that an assignment ${\mathbf x}\wedge {\mathbf y} = (x_1, \cdots, x_n, 0, x_{n+2}, \cdots, x_{2n-1})$ also satisfies $\Xi$, and therefore, is in $\Xi^{\boldsymbol{\Gamma}}\cup \delta$. Thus,  $\Xi^{\boldsymbol{\Gamma}}\cup \delta\in \langle \Gamma\rangle$ and $\rho$ is strongly reducible to $\Gamma$.
\end{proof}

\begin{proof}[Theorem~\ref{Horn-case-reduce}] Let $\rho\in \langle \Gamma\rangle$ be $n$-ary, i.e. $\rho$ is closed with respect to component-wise conjunction. A classical result about
$\wedge$-closed relations (see~\cite{Alfred,McKinsey}) states that $\rho$ can be represented as:
$$
\rho = \bigcap_{i=1}^l \rho_i
$$
where $\rho_i = \{(x_1, \cdots, x_n) \mid \Phi_i(x_{s_{i1}}, \cdots, x_{s_{i r_i}})\}$ where $\Phi_i$ is a Horn clause. From the previous Lemma we conclude that each of $\rho_i, i\in [l]$ is strongly reducible to $\Gamma$. Therefore, $\rho$ is A-reducible to $\Gamma$. Since this is true for any $\rho\in \langle \Gamma\rangle$, we conclude that $\Gamma$ is an A-language.
\end{proof}
\section{Proof of Theorem~\ref{Horn-case}}\label{proof-of-horn}
In this case we have a vocabulary $\tau=\{\pi_1, \pi_2, \pi_3\}$ where $\pi_1, \pi_2$ are unary and $\pi_3$ is assigned an arity 3. 

Let ${\mathbf R} = (V, Z, O, r)$ be an instance of ${\rm Dense}(\boldsymbol{\Gamma})$. Let us define an implicational system $\Sigma$ on $V$ that consists of rules $\{i,j\}\rightarrow k$ for any $(i,j,k)\in r$. The implicational system $\Sigma$ defines a closure operator $o_{\Sigma}(S) = \{x| (S\rightarrow x)\in \Sigma^\triangleright\}$.
Let ${\mathbf R}' = (V, Z', O', r')$ be a maximal instance such that $Z'\supseteq Z$, $O'\supseteq O$, $r'\supseteq r$ and ${\rm Hom} ({\mathbf R}, \boldsymbol{\Gamma})={\rm Hom} ({\mathbf R}', \boldsymbol{\Gamma})\ne \emptyset$. Note that $(i,j,k)\in r'$ if and only if $k\in o_\Sigma (\{i,j\}\cup O)$ and $Z\cap o_\Sigma (\{i,j\}\cup O) = \emptyset$. 
Indeed, for any $k\in o_\Sigma (\{i,j\}\cup O)$ we have $(i,j,k)\in r'$, because $\{i,j\}\cup O\rightarrow k$ is a consequence of rules in $r$. On the contrary, let $k\notin o_\Sigma (\{i,j\}\cup O)$. Then, $h: V\to D$ defined by $h(v) = 1$ if $v\in o_\Sigma (\{i,j\}\cup O)$ and $h(v) = 0$, if otherwise, is a homomorphism from ${\mathbf R}$ to $\boldsymbol{\Gamma}$. Therefore, for any $k\notin o_\Sigma (\{i,j\}\cup O)$ we have $(h(i), h(j), h(k))\notin \varrho_3$. Using Theorem~\ref{jvm}, we obtain $(i,j,k)\notin r'$.

Thus, for any $(i,j,k)\in r'$ there exists a derivation of $k$ from $\{i,j\}\cup O$ using only rules $\{i,j\}\rightarrow k$, $(i,j,k)\in r$. To such a derivation one can always correspond a rooted binary tree $T$ whose nodes are labeled with elements of $V$, the root is labeled with $k$, and all leaves are labeled by elements of $\{i,j\}\cup O$. Any (non-leaf) node $p$ (a parent) of the tree $T$ has two children $c_1, c_2$ such that $\{l(c_1), l(c_2)\}\rightarrow l(p)$ is in $\Sigma$ ($l$ is a labeling function).

Let $x,y$ be two leaves of the tree $T$ with a common parent $z$ such that the distance from $x$ to the root $k$ equals the depth of the tree (i.e. is the largest possible one). The parent of $z$ is denoted by $u$ and all possible branches under $u$ are drawn in Figure~\ref{fig:somthing1}: we reduced the number of possible branches to analyze using the rule $\pi_3(x,y,u) \to \pi_3 (y,x,u)$ that makes an order of children irrelevant. Circled leaves correspond to leaves labeled by elements of $O$.  A leaf that is not circled can be labeled either by $i,j$ or by an element from $O$.
For each case, the Figure shows how to reduce the tree $T$ by deleting redundant nodes under $u$. To delete the redundant nodes and connect leaves to $u$ we have to verify that a new reduced branch with a parent $u$ and 2 leaves $x, y$ (or, $x, t$) corresponds to a triple $(x,y,u)\in r^L$ (or, $(x,t,u)\in r^L$), i.e. the resulting triple can be obtained using rules from $L$. 
Needed rules are indicated near each deletion operation in Figure~\ref{fig:somthing1}. 

It is easy to see that using such deletions we will eventually obtain a root $k$ with two children labeled by $c_1,c_2\in \{i,j\} \cup O$. Therefore, the triple $(c_1, c_2, k)$ is in $r^L$. If $\{c_1, c_2\} = \{i,j\}$, then $(i,j,k)$ can be obtained from $(c_1, c_2, k)$ using the rule (1) from the list below. If $c_1=i$ and $c_2\in O$ (or, $c_1, c_2\in O$), then $(i,j,k)$ can be obtained from $(c_1, c_2, k)$ using the rule (2). 
Thus, $(i,j,k)\in r^L$, i.e. we proved that $r' = r^L$.

Let us show now that $O'=O^L$. Analogously to the previous analysis, $k\in o_\Sigma(O)$ if there is a derivation tree with a root $k$ labeled with elements of $V$ and all leaves are labeled by elements of $O$. Using the same reduction we finally obtain the triple $(i,j,k)\in r^L$, where $i,j\in O$. Using the rule (3), we conclude $k\in O^L$, i.e. we proved the inclusion $O^L \supseteq o_\Sigma(O)$. Therefore, $O^L = o_\Sigma(O)$. Then, $h: V\to D$ defined by $h(v) = 1$ if $v\in o_\Sigma (O)$ and $h(v) = 0$, if otherwise, is a homomorphism from ${\mathbf R}$ to $\boldsymbol{\Gamma}$. Since for any $v\notin O^L$ we have $h(v)\notin \varrho_2$, then using Theorem~\ref{jvm}, we obtain that $o_\Sigma (O) = O^L$ is maximal and $O'=O^L$.

Finally, let us prove that $Z'=Z^L$. First, let us prove $Z' = \{v\in V| o_\Sigma (\{v\}\cup O)\cap Z \ne \emptyset\}$. Indeed, if $a\in V$ is such that $o_\Sigma (\{a\}\cup O)\cap Z \ne \emptyset$, then the set $\{h\in {\rm Hom} ({\mathbf R}, \boldsymbol{\Gamma}) | h(a)=1\}$ is empty. Therefore, $h(a)=0$ for any $h\in {\rm Hom} ({\mathbf R}, \boldsymbol{\Gamma})$, which implies $a\in Z'$. On the contrary, if $a\in V$ is such that $o_\Sigma (\{a\}\cup O)\cap Z = \emptyset$, then $h: V\to D$ defined by $h(v) = 1$ if $v\in o_\Sigma (\{a\}\cup O)$ and $h(v) = 0$, if otherwise, is a homomorphism from ${\mathbf R}$ to $\boldsymbol{\Gamma}$. Therefore, $a\notin Z'$.

Thus, $Z'$ is a set of all elements $a\in V$ such that some element $r\in Z$ can be derived from $\{a\}\cup O$ in the implicational system $\Sigma$. Analogously to the previous case, there is a rooted binary tree $T$ with a root $r\in Z$ whose nodes are labeled by elements of $V$ and leaves are labeled by $\{a\}\cup O$. Using the same technique this tree can be reduced to a root $r$ with two children $c_1$ and $c_2$, such that $\{c_1, c_2\}\subseteq \{a\} \cup O$, $\{c_1, c_2\}\not\subseteq  O$ and $(c_1, c_2, r)\in r^L$. W.l.o.g. let $c_1=a$. If $c_2\in O$, then using the rule (4) we can deduce $a\in Z^L$. 
If $c_2=a$, then using the rule (5) we can deduce $a\in Z^L$. 
Thus, $Z'\subseteq Z^L$, and consequently, $Z'= Z^L$. 

In the case ${\rm Hom} ({\mathbf R}, \boldsymbol{\Gamma})= \emptyset$, it is easy to see that we will eventually apply the rule (6).
\begin{figure}[htbp]
    \centering
        \includegraphics[width=0.6\textwidth]{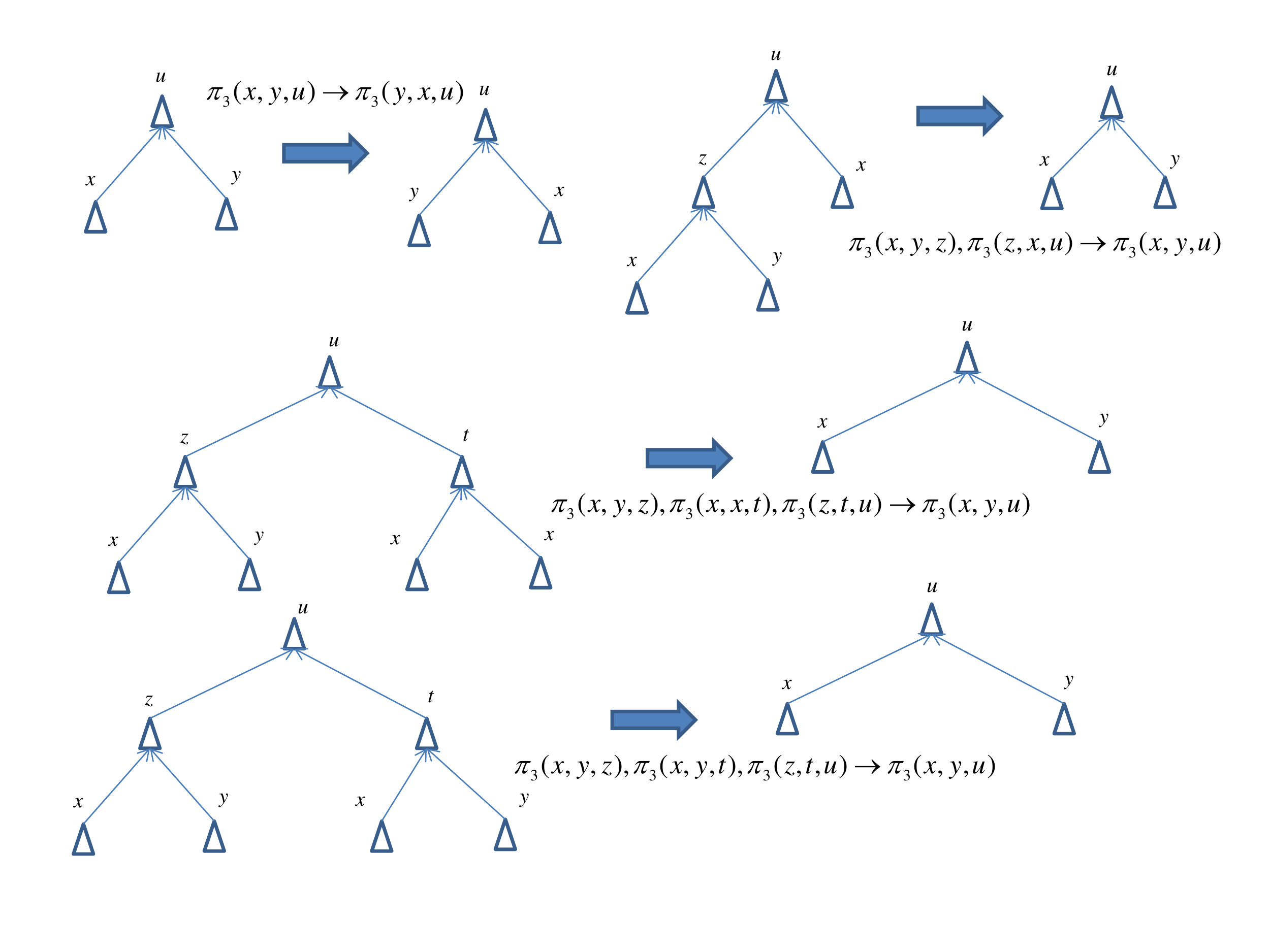}
        \includegraphics[width=0.6\textwidth]{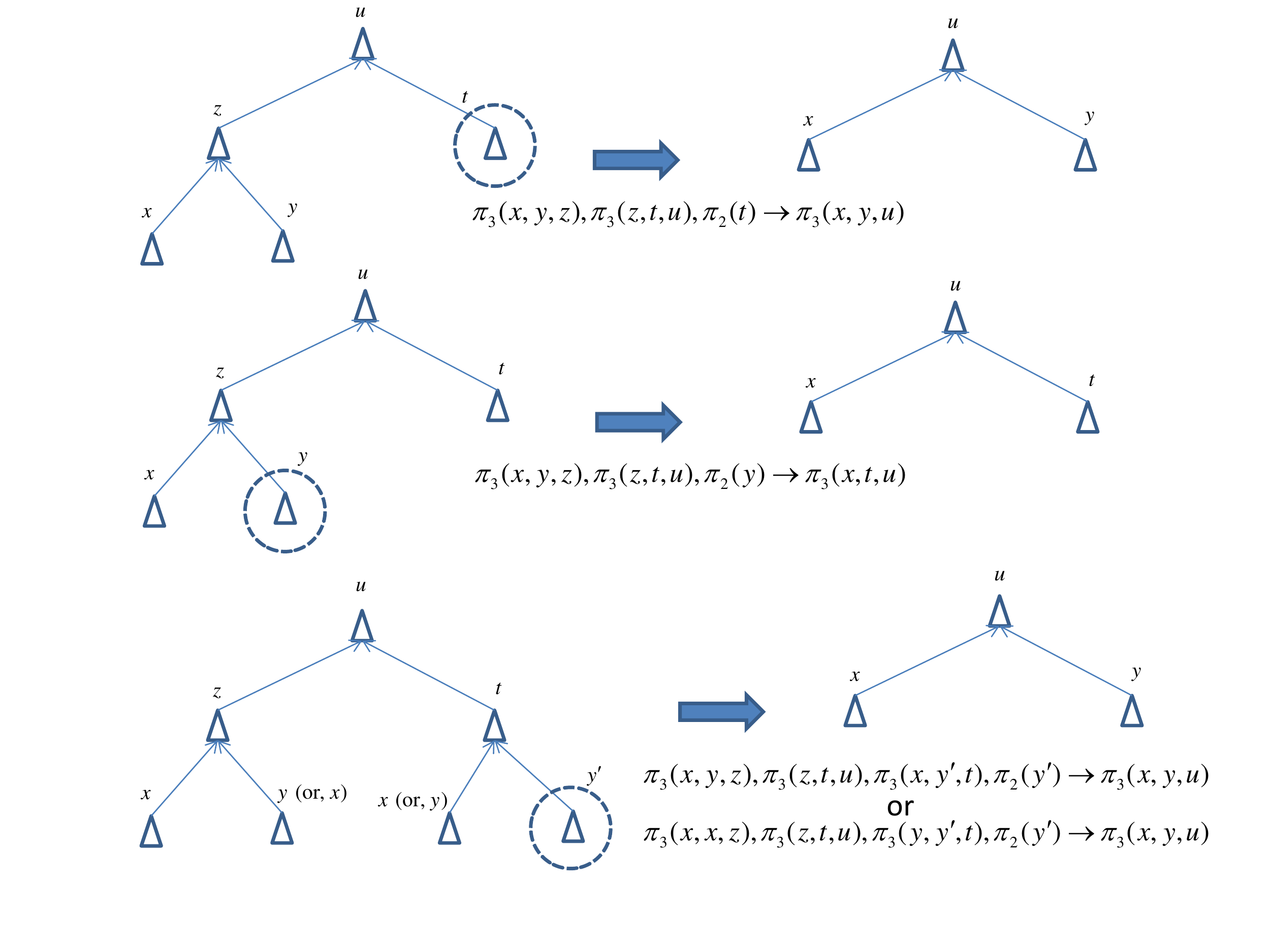}
        \includegraphics[width=0.6\textwidth]{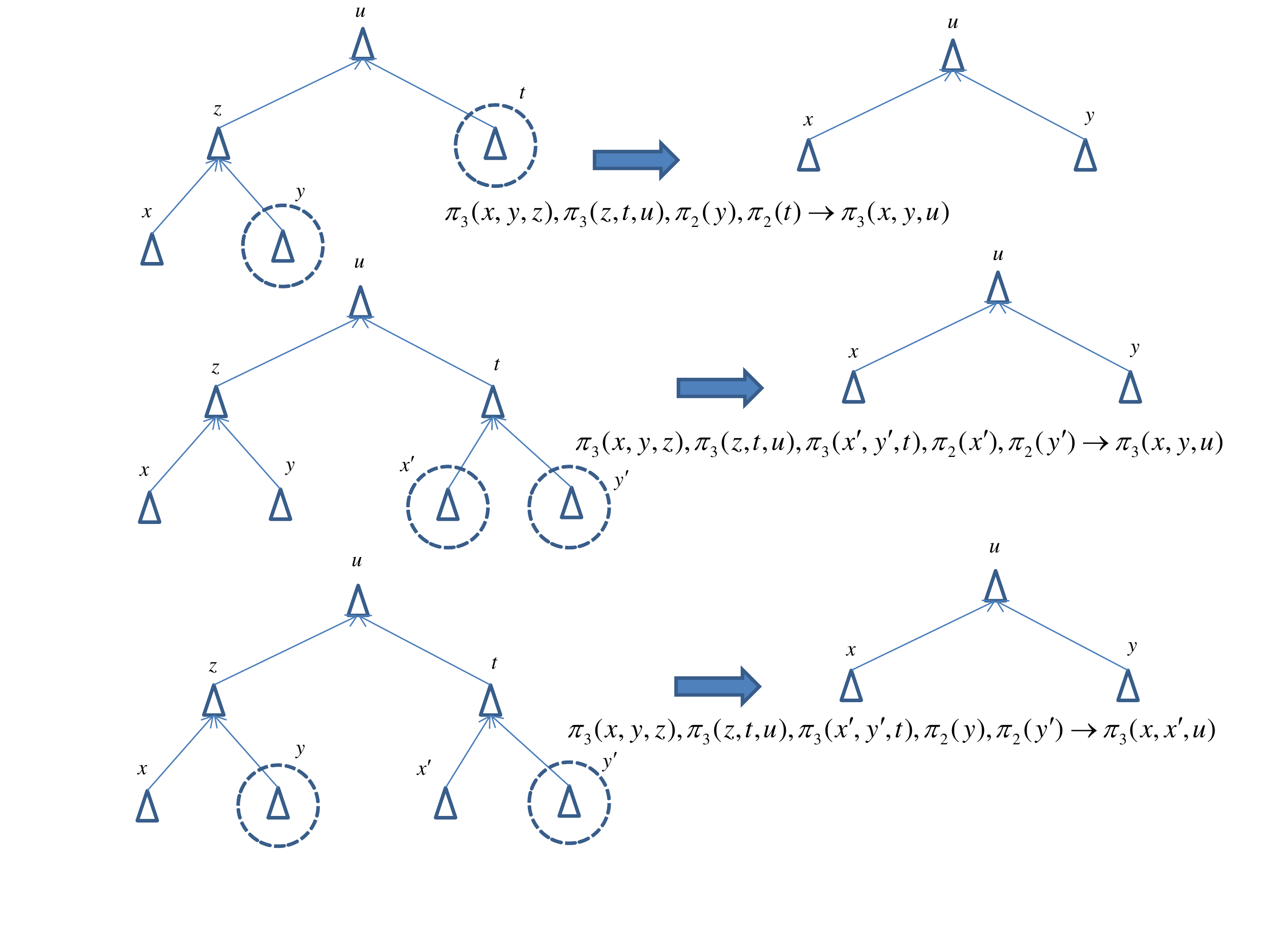}
    \caption{A new reduced branch with a parent $u$ and 2 leaves $x, y$ (or, $x, t$) corresponds to a triple $(x,y,u)\in r^L$. There is no need to list cases with 3 nodes labeled by $O$, because they all are subcases of the listed.}
    \label{fig:somthing1}
\end{figure}
The complete list of Horn formulas in $L$ is given below:
\begin{itemize}
\item[(1)] $ \forall x,y,u\,\, \big(\pi_3 (x, y, u) \rightarrow \pi_3 (y, x, u)\big) $
\item[(2)] $ \forall x,y,z,u\,\, \big(\pi_3 (x, y, u)\wedge \pi_2 (x) \rightarrow \pi_3 (z, y, u)\big) $
\item[(3)] $ \forall x,y,z,u\,\, \big(\pi_3 (x, y, u)\wedge \pi_2 (x)\wedge \pi_2 (y) \rightarrow \pi_2 (u)\big) $
\item[(4)] $ \forall x,y,z,u\,\, \big(\pi_3 (x, y, u)\wedge \pi_2 (x) \wedge \pi_1 (u) \rightarrow \pi_1 (y)\big) $
\item[(5)] $ \forall x,y\,\, \big(\pi_3 (x, x, y)\wedge \pi_1 (y) \rightarrow \pi_1 (x)\big) $
\item[(6)] $ \forall x\,\, \big(\pi_1 (x)\wedge \pi_2 (x) \rightarrow {\rm F}\big) $
\item[(7)] $ \forall x,y,z,u \big(\pi_3 (x, y, z) \wedge \pi_3 (z, x, u) \rightarrow \pi_3 (x, y, u)\big) $
\item[(8)] $ \forall x,y,z,t,u \big(\pi_3 (x, y, z) \wedge \pi_3 (x, x, t) \wedge  \pi_3 (z, t, u) \rightarrow \pi_3 (x, y, u)\big) $
\item[(9)] $ \forall x,y,z,t,u \big(\pi_3 (x, y, z) \wedge \pi_3 (x, y, t) \wedge  \pi_3 (z, t, u) \rightarrow \pi_3 (x, y, u)\big) $
\item[(10)] $ \forall x,y,z,t,u \big(\pi_3 (x, y, z)\wedge \pi_3 (z, t, u)\wedge \pi_2(t) \rightarrow \pi_3(x,t,u)\big) $
\item[(11)] $ \forall x,y,z,t,u \big(\pi_3 (x, y, z)\wedge \pi_3 (z, t, u)\wedge \pi_2(y) \rightarrow \pi_3(x,y,u)\big) $
\item[(12)] $ \forall x,y,y',z,t,u \big(\pi_3 (x, y, z)\wedge \pi_3 (z, t, u)\wedge  \pi_3 (x, y', t)\wedge \pi_2(y') \rightarrow \pi_3(x,y,u)\big) $
\item[(13)] $ \forall x,y,y',z,t,u \big(\pi_3 (x, x, z)\wedge \pi_3 (z, t, u)\wedge \pi_3 (y, y', t)\wedge  \pi_2(y') \rightarrow \pi_3(x,y,u)\big) $
\item[(14)] $ \forall x,y,z,t,u \big(\pi_3 (x, y, z)\wedge \pi_3 (z, t, u)\wedge \pi_2 (y)\wedge  \pi_2(t) \rightarrow \pi_3(x,y,u)\big) $
\item[(15)] $ \forall x,y,x',y',z,t,u \big(\pi_3 (x, y, z)\wedge \pi_3 (z, t, u)\wedge 
\pi_3 (x', y', t)\wedge \pi_2 (x')\wedge \pi_2(y') \rightarrow \pi_3(x,y,u)\big) $
\item[(16)] $ \forall x,y,x',y',z,t,u \big(\pi_3 (x, y, z)\wedge \pi_3 (z, t, u)\wedge 
\pi_3 (x', y', t)\wedge \pi_2 (y)\wedge \pi_2(y') \rightarrow \pi_3(x,x',u)\big) $
\end{itemize}
This list is not optimized and some formulas could be derivable from others.

\section{Proof of Theorem~\ref{2-SAT}.} \label{2-SAT-proof}
Throughout the proof we assume $D=\{0,1\}$ and $\boldsymbol{\Gamma} = (D, \varrho_1, \varrho_2, \varrho_3)$ where $ \varrho_1 = \big \{ ({x, y} ) | x \vee y \big \} $,
 $ \varrho_2 = \big \{ ({x, y} ) | \neg x \vee y \big \} $
 and $ \varrho_3 = \big \{{ ({x, y} ) | \neg x \vee \neg y} \big \} $.
For $\rho_1, \rho_2\subseteq D^2$ let us denote
$$
\rho_1\circ \rho_2 = \{(x,z)| \exists y: (x,y)\in\rho_1 {\rm\,\,and\,\,} (y,z)\in \rho_2\}
$$
\begin{definition} Let $\Gamma_2$ be a set of all nonempty binary relations over $D$. A subset $C\subseteq {\mathcal C}^{\boldsymbol{\Gamma}_2}_V$ is called full if for any $u,v\in V$ there exists only one $\langle (u,v), \rho \rangle\in C$. A full subset $C\subseteq {\mathcal C}^{\boldsymbol{\Gamma}_2}_V$ is called path-consistent if for any $\langle (u,v), \rho_1 \rangle, \langle (v,w), \rho_2 \rangle, \langle (u,w), \rho_3 \rangle\in C$ we have $\rho_3\subseteq \rho_1\circ \rho_2$ and for any $\langle (u,u), \rho \rangle\in C$ we have $\rho\subseteq \{(a,a)|a\in D\}$.
\end{definition}
It is well-known that for binary constraint satisfaction problems, path consistency is equivalent to 3-local consistency~\cite{Dechter}. Therefore, if $C\subseteq {\mathcal C}^{\boldsymbol{\Gamma}_2}_V$ is path-consistent, then the corresponding 2-SAT instance is satisfiable.

Let us introduce the set of formulas:
\begin{enumerate}
\item  $\forall x\,\,{\rm True} \rightarrow \pi_2(x, x)$ \label{seventh}
\item $\forall x,y\,\, \big(\pi_1(x, y) \rightarrow \pi_1(y, x)\big)$ \label{first}
\item $\forall x,y\,\, \big(\pi_3(x, y) \rightarrow \pi_3(y, x)\big)$ \label{second}
\item $\forall x,y,z\,\, \big(\pi_2(x, y) \wedge \pi_2(y, z)  \rightarrow \pi_2(x, z)\big)$ \label{third}
\item  $\forall x,y,z\,\, \big(\pi_1 (x, y) \wedge \pi_2(y, z)  \rightarrow  \pi_1(x, z)\big)$ \label{fourth}

\item  $\forall x,y,z\,\, \big(\pi_3 (x, y) \wedge \pi_2 (z, y)  \rightarrow \pi_3 (x, z)\big)$ \label{fifth}

\item  $\forall x,y,z\,\, \big(\pi_3 (x, y) \wedge \pi_1(y, z) \rightarrow \pi_2(x, z)\big)$ \label{sixth}
\end{enumerate}

To any relational structure ${\mathbf R} = (V, r_1, r_2, r_3)$, where $r_i, i\in [r]$ is a binary relation, one can correspond the full subset:
$$
C({\mathbf R}) = \{\langle (u,v), \rho_{uv} \rangle | u,v\in V\} \subseteq {\mathcal C}^{\boldsymbol{\Gamma}_2}_V
$$
where
\begin{equation*}
\begin{split}
\rho_{uv} = \bigcap_{i: (u,v)\in r_i}\varrho_i \bigcap_{i: (u,v)\in r^T_i}\varrho^T_i, {\rm\,\, if\,\,} u\ne v \\
\rho_{uu} = \bigcap_{i: (u,u)\in r_i}\varrho_i \bigcap_{i: (u,u)\in r^T_i}\varrho^T_i \cap \{(a,a)|a\in D\}
\end{split}
\end{equation*}

\begin{lemma}\label{path} If ${\mathbf R} = (V, r_1, r_2, r_3)$ satisfies the formulas 1-7 and $r_1\cap r_2\cap r_3\cap r^T_2=\emptyset$, $r_1\cap r_3\cap \{(u,u)|u\in V\}=\emptyset$, then $C({\mathbf R})$ is path-consistent.
\end{lemma}
\begin{proof} Properties~\ref{first} and~\ref{second} claim that $r_1$ and $r_3$ are symmetric relations, therefore we have $r_1 = r^T_1$ and $r_3 = r^T_3$.
 Since $r_1\cap r_2\cap r_3\cap r^T_2=\emptyset$, then the set $\{\varrho_i| (u,v)\in r_i\}\cup \{\varrho^T_i| (u,v)\in r^T_i\} \ne \{\varrho_1,\varrho_2,\varrho_3,\varrho^T_2\}$ for any $(u,v)$. Since $\bigcap_{a\in A} a\ne \emptyset$ for any proper subset $A\subset \{\varrho_1,\varrho_2,\varrho_3,\varrho^T_2\}$, then $\rho_{uv}\ne \emptyset$ for any $u\ne v$. 

Due to the property~\ref{seventh}, we have $(u,u)\in r_2\cap r_2^T$ for any $u\in V$. Also, $(u,u)\notin r_1\cap r_3$ because of $r_1\cap r_3\cap \{(v,v)|v\in V\}=\emptyset$. Therefore,
for any $u\in V$, the set $\{\varrho_i| (u,u)\in r_i\}\cup \{\varrho^T_i| (u,u)\in r^T_i\}$ is a proper subset of $\{\varrho_1,\varrho_3\}$. Thus, $\rho_{uu}\ne \emptyset$ and $\rho_{uu} \subseteq \{(a,a)| a\in D\}$.

Note that for any $u \ne v$: a) $(0,0)\notin \rho_{uv}$ if and only if $(u,v)\in r_1$, b) $(1,1)\notin \rho_{uv}$ if and only if $(u,v)\in r_3$, c) $(1,0)\notin \rho_{uv}$ if and only if $(u,v)\in r_2$, and d) $(0,1)\notin \rho_{uv}$ if and only if $(v,u)\in r_2$.

Let us prove that $\rho_{uw}\subseteq \rho_{uv}\circ \rho_{vw}$ for any $u,v,w\in V$. 
Let us first consider the case of distinct $u,v,w$.
Let $(a,c)\in \rho_{uw}$. Our goal is to show that there exists $b$ such that $(a,b)\in \rho_{uv}$ and $(b,c)\in \rho_{vw}$. Let us prove the last statement by reductio ad absurdum. Assume that for any $b$ we have $(a,b)\notin \rho_{uv}$, $(b,c)\notin \rho_{vw}$ and $(a,c)\in \rho_{uw}$.

There are 4 possibilities for $(a,c)$: $(0,0)$, $(1,1)$, $(0,1)$ and $(1,0)$. Let us list all of them and check that $(a,b)\notin \rho_{uv}$ and $(b,c)\notin \rho_{vw}$ and $(a,c)\in \rho_{uw}$ cannot hold for any $b\in \{0,1\}$.

The case $(a,c) = (0,0)$: $(0,b)\notin \rho_{uv}$ and $(b,0)\notin \rho_{vw}$ for $b\in \{0,1\}$ implies $(u,v)\in r_1\cap r^T_2$ and $(v,w)\in r_1\cap r_2$. Due to the property~\ref{fourth} we have $(u,w)\in r_1$ and this contradicts to $(0,0)\in \rho_{uw}$.

The case $(a,c) = (1,1)$: $(1,b)\notin \rho_{uv}$ and $(b,1)\notin \rho_{vw}$ for $b\in \{0,1\}$ implies $(u,v)\in r_3\cap r_2$ and $(v,w)\in r_3\cap r^T_2$. Due to the property~\ref{fifth} we have $(u,w)\in r_3$ and this contradicts to $(1,1)\in \rho_{uw}$.

The case $(a,c) = (0,1)$: $(0,b)\notin \rho_{uv}$ and $(b,1)\notin \rho_{vw}$ for $b\in \{0,1\}$ implies $(u,v)\in r_1\cap r^T_2$ and $(v,w)\in r_3\cap r^T_2$. Due to the property~\ref{third} we have $(w,u)\in r_2$ and this contradicts to $(0,1)\in \rho_{uw}$.

The case $(a,c) = (1,0)$: $(1,b)\notin \rho_{uv}$ and $(b,0)\notin \rho_{vw}$ for $b\in \{0,1\}$ implies $(u,v)\in r_3\cap r_2$ and $(v,w)\in r_1\cap r_2$. Due to the property~\ref{third} we have $(u,w)\in r_2$ and this contradicts to $(1,0)\in \rho_{uw}$.

It remains to check path-consistency property for any triple of variables $u, v, w\in V$ where either $u=w\ne v$ or $u=v\ne w$ (i.e. 2-local consistency). The case $u=v=w$ is trivial.

Let us check the case $u=w\ne v$. Let $(a,a)\in \rho_{uu}$. Let us assume that for any $b\in D$ we have $(a,b)\notin \rho_{uv}$. The case $a=0$ gives $(0,0)\in \rho_{uu}$, $(0,0),(0,1)\notin \rho_{uv}$, and therefore, $(u,u)\notin r_1$,
$(u,v)\in r_1\cap r_2^T$. From property~\ref{fourth} we conclude $(u,u)\in r_1$ and obtain a contradiction. The case $a=1$ gives $(1,1)\in \rho_{uu}$, $(1,0),(1,1)\notin \rho_{uv}$, and therefore, $(u,u)\notin r_3$,
$(u,v)\in r_3\cap r_2$. From property~\ref{fifth} we conclude $(u,u)\in r_3$ and obtain a contradiction.

Finally, let us check the case $u=v\ne w$. Let $(a,c)\in \rho_{uw}$ and for any $b\in D$ we have $(a,b)\notin \rho_{uu}, (b,c)\notin \rho_{uw}$. 

The case $(a,c)=(0,0)$ gives $(0,0)\in \rho_{uw}$, $(0,b)\notin \rho_{uu}, (b,0)\notin \rho_{uw}$. The last is equivalent to $(u,w)\notin r_1$, $(u,u)\in r_1$, $(u,w)\in r_1\cap r_2$. From property~\ref{fourth} we conclude $(u,w)\in r_1$ and obtain a contradiction.

The case $(a,c)=(1,1)$ gives $(1,1)\in \rho_{uw}$, $(1,b)\notin \rho_{uu}, (b,1)\notin \rho_{uw}$. The last is equivalent to $(u,w)\notin r_3$, $(u,u)\in r_3$, $(u,w)\in r_3\cap r^T_2$. From property~\ref{fifth} we conclude $(u,w)\in r_3$ and obtain a contradiction.

The case $(a,c)=(0,1)$ gives $(0,1)\in \rho_{uw}$, $(0,b)\notin \rho_{uu}, (b,1)\notin \rho_{uw}$. The last is equivalent to $(u,w)\notin r^T_2$, $(u,u)\in r_1$, $(u,w)\in r_3\cap r^T_2$. From property~\ref{sixth} we conclude $(w,u)\in r_2$ and obtain a contradiction.

The case $(a,c)=(1,0)$ gives $(1,0)\in \rho_{uw}$, $(1,b)\notin \rho_{uu}, (b,0)\notin \rho_{uw}$. The last is equivalent to $(u,w)\notin r_2$, $(u,u)\in r_3$, $(u,w)\in r_1\cap r_2$. From property~\ref{sixth} we conclude $(u,w)\in r_2$ and obtain a contradiction. Thus, the lemma is proved.
\end{proof}

\begin{corollary} Let $L$ be the set of formulas 1-7 and $L^{\rm stop} = \{\pi_1(x,y)\wedge \pi_2(x,y)\wedge \pi_3(x,y) \wedge \pi_2(y,x)\rightarrow {\rm F}, \pi_1(x,x)\wedge \pi_3(x,x) \rightarrow {\rm F}\}$. Then, ${\rm Dense}(\boldsymbol{\Gamma})$ can be solved by the Datalog program $L\cup L^{\rm stop}$.
\end{corollary}
\begin{proof} Let ${\mathbf R}$ be an instance of ${\rm Dense}(\boldsymbol{\Gamma})$. If ${\rm Hom}({\mathbf R}, \boldsymbol{\Gamma})=\emptyset$, then ${\rm Hom}({\mathbf R}^L, \boldsymbol{\Gamma})=\emptyset$. By construction, ${\mathbf R}^L$ satisfies properties 1-7. If $r^L_1\cap r^L_2\cap r^L_3 \cap (r^L_2)^T = \emptyset$ and $r^L_1\cap r^L_3\cap \{(v,v)|v\in V\}=\emptyset$, then, by Lemma~\ref{path}, the subset $C({\mathbf R}^L)$ is path-consistent (and therefore, is satisfiable). The last contradicts to ${\rm Hom}({\mathbf R}^L, \boldsymbol{\Gamma})=\emptyset$. Therefore, either $r^L_1\cap r^L_2\cap r^L_3 \cap (r^L_2)^T \ne \emptyset$ or $r^L_1\cap r^L_3\cap \{(v,v)|v\in V\}\ne \emptyset$. In that case the Datalog program will identify the emptiness of ${\rm Hom}({\mathbf R}, \boldsymbol{\Gamma})$ by applying the rule $\pi_1(x,y)\wedge \pi_2(x,y)\wedge \pi_3(x,y) \wedge \pi_2(y,x)\rightarrow {\rm F}$ to $(u,v)\in r^L_1\cap r^L_2\cap r^L_3 \cap (r^L_2)^T$ or the rule $\pi_1(x,x)\wedge \pi_3(x,x) \rightarrow {\rm F}$ to $(u,u)\in r^L_1\cap r^L_3\cap \{(v,v)|v\in V\}$.

Let us now consider the case ${\rm Hom}({\mathbf R}^L, \boldsymbol{\Gamma})\ne\emptyset$. In that case we have $r^L_1\cap r^L_2\cap r^L_3 \cap (r^L_2)^T = \emptyset$, $r^L_1\cap r^L_3\cap \{(v,v)|v\in V\}=\emptyset$ and the subset $C({\mathbf R}^L)$ is path-consistent. A well-known application of Baker-Pixley theorem to languages with a majority polymorphism~\cite{JEAVONS1998251} gives us that path-consistency (or, 3-consistency) implies global consistency. Thus, any 3-consistent solution can be globally extended, i.e.
$$
{\rm pr}_{u,v}  {\rm Hom}({\mathbf R}, \boldsymbol{\Gamma}) = {\rm pr}_{u,v}  {\rm Hom}({\mathbf R}^L, \boldsymbol{\Gamma}) = \rho_{u,v}
$$
for any $\langle (u,v), \rho_{uv}\rangle \in C({\mathbf R}^L)$. Thus, 
$$
\bigcap_{h\in {\rm Hom}({\mathbf R}, \boldsymbol{\Gamma})} h^{-1}(\varrho_i)  = \{(u,v)| {\rm pr}_{u,v}  {\rm Hom}({\mathbf R}, \boldsymbol{\Gamma})\subseteq \varrho_i\} = \{(u,v)| \rho_{u,v}\subseteq \varrho_i\} \subseteq r^L_i
$$
The last implies that $({\mathbf R}^L, \boldsymbol{\Gamma})$ is a maximal pair, and this completes the proof.
\end{proof}
\section{Conclusion and open questions}
We studied the size of an implicational system $\Sigma$ corresponding to a densification operator on a set of constraints for different constraint languages. It turns out that only for bounded width languages this size can be bounded by a polynomial of the number of variables. This naturally led us to more efficient algorithms for the densification and the sparsification tasks.

An unresolved issue of the paper is a relationship (equality?) between the following classes of constraint languages: a) core languages with a weak polynomial densification operator, b) core languages of bounded width. 
Also, the complexity classification of ${\rm Dense}(\boldsymbol{\Gamma})$ for the general domain $D$ is still open.

\bibliographystyle{unsrt}

\end{document}